\documentclass[sigconf, screen=true, review=false, printacmref=false, printccs=false, printfolios=false]{acmart}
\usepackage{comment}
\usepackage[algoruled,boxed,lined]{algorithm2e}
\usepackage{cleveref}
\DeclareMathAlphabet\mathbfcal{OMS}{cmsy}{b}{n}

\setcopyright{none}
\usepackage{booktabs} 
\usepackage{epstopdf}
\usepackage{comment}
\usepackage{tabularx}
\usepackage{subfigure}
\usepackage{paralist}
\usepackage{balance}

\usepackage{xcolor}
\definecolor{thedarkblue}{RGB}{0,0,120} 
\definecolor{mydarkblue}{rgb}{0,0.08,0.45} 
\definecolor{darkblue}{rgb}{0,0.08,180}
\colorlet{TufteRed}{red!80!black}

\definecolor{theblue}{RGB}{0,0,180}
\colorlet{thered}{TufteRed}

\usepackage{hyperref}
\hypersetup{
colorlinks=true,
linkcolor=mydarkblue,
citecolor=mydarkblue,
filecolor=mydarkblue,
urlcolor=mydarkblue}

\usepackage{microtype}
\usepackage{balance}
\usepackage{amsmath,amssymb,amsthm}

\usepackage{tikz}
\usepackage{verbatim}
\usetikzlibrary{arrows}
\usetikzlibrary{shapes,snakes}
\usetikzlibrary{decorations.pathmorphing} 
\usetikzlibrary{fit}					
\usetikzlibrary{backgrounds}	

\usepackage{ragged2e}
\usepackage{multirow}
\usepackage{microtype}
\usepackage{balance}
\usepackage{setspace}

\graphicspath{{./}{./graphics/}}
\newcolumntype{H}{>{\setbox0=\hbox\bgroup}c<{\egroup}@{}}

\newcolumntype{R}[1]{>{\RaggedLeft\arraybackslash}} 
\newcolumntype{L}[1]{>{\RaggedRight\arraybackslash}}

\newcommand{\abs}[1]{\left|#1\right|}

\newcommand{\eg}{\emph{e.g.}}
\newcommand{\ie}{\emph{i.e.}}

\newtheorem{theorem}{\bfseries{Theorem}}
\newtheorem{corollary}{\bfseries{Corollary}}

\newtheorem*{note}{\hspace{-1em}\textrm{Note}}

\newtheorem{lemma}{\bfseries{Lemma}}
\newtheorem{Definition}{\hspace{-1em}\bfseries{Definition}}

\providecommand{\tensor}[1]{\boldsymbol{\mathcal{#1}}}
\providecommand{\mat}[1]{\boldsymbol{\mathrm{#1}}}
\renewcommand{\vec}[1]{\boldsymbol{\mathrm{#1}}}

\DeclareMathOperator*{\argmin}{argmin}

\DeclareMathOperator{\hugeE}{\mbox{\huge\raise-0.3ex\hbox{E}}}
\DeclareMathOperator{\p}{\mathbb{P}}
\DeclareMathOperator{\hugep}{\mbox{\huge\raise-0.3ex\hbox{$\p$}}}

\newcommand{\RR}{\mathbb{R}}

\providecommand{\eye}{\mat{I}}

\providecommand{\mD}{\ensuremath{\mat{D}}}

\providecommand{\mI}{\ensuremath{\mat{I}}}

\providecommand{\mL}{\ensuremath{\mat{L}}}

\providecommand{\mW}{\ensuremath{\mat{W}}}

\providecommand{\mY}{\ensuremath{\mat{Y}}}
\providecommand{\mZ}{\ensuremath{\mat{Z}}}

\providecommand{\tX}{\ensuremath{\tensor{X}}}

\providecommand{\vv}{\ensuremath{\vec{v}}}

\providecommand{\vx}{\ensuremath{\vec{x}}}

\providecommand{\vz}{\ensuremath{\vec{z}}}

\DeclareMathOperator{\cut}{cut}
\DeclareMathOperator{\vol}{vol}

\begin{document}

\title{Higher-order Spectral Clustering for Heterogeneous Graphs}

\author{Aldo G. Carranza}
\affiliation{
\institution{Stanford University}
}

\author{Ryan A. Rossi}
\orcid{1234-5678-9012-3456}
\affiliation{
\institution{Adobe Research}
}

\author{Anup Rao}
\affiliation{
\institution{Adobe Research}
}

\author{Eunyee Koh}
\affiliation{
\institution{Adobe Research}
}
\email{} 

\renewcommand\shortauthors{Carranza, A.G. et al.}

\begin{abstract}
Higher-order connectivity patterns such as small induced sub-graphs called graphlets (network motifs) are vital to understand the important components (modules/functional units) governing the configuration and behavior of complex networks.
Existing work in higher-order clustering has focused on simple \emph{homogeneous graphs} with a single node/edge type.
However, heterogeneous graphs consisting of nodes and edges of different types are seemingly ubiquitous in the real-world.
In this work, we introduce the notion of typed-graphlet that explicitly captures the rich (typed) connectivity patterns in heterogeneous networks.
Using typed-graphlets as a basis, we develop a general principled framework for higher-order clustering in heterogeneous networks.
The framework provides mathematical guarantees on the optimality of the higher-order clustering obtained.
The experiments demonstrate the effectiveness of the framework quantitatively for three important applications including
(i) clustering,
(ii) link prediction, and 
(iii) graph compression.
In particular, the approach achieves a mean improvement of 43x over all methods and graphs for clustering while achieving a $18.7\%$ and $20.8\%$ improvement for link prediction and graph compression, respectively.
\end{abstract}



\maketitle

\section{Introduction} \label{sec:intro}
Clustering in graphs has been one of the most fundamental tools for analyzing and understanding the components of complex networks.
It has been used extensively in many important applications to distributed systems~\cite{hendrickson1995improved, simon1991partitioning, van1995improved}, compression~\cite{rossi2015pmc-sisc, buehrer2008scalable}, image segmentation~\cite{shi2000normalized, felzenszwalb2004efficient}, document and word clustering \cite{dhillon2001co}, among others.
Most clustering methods focus on simple flat/\emph{homogeneous} graphs where nodes and edges represent a single entity and relationship type, respectively. 
However, heterogeneous graphs consisting of nodes and edges of different types are seemingly ubiquitous in the real-world.
In fact, most real-world systems give rise to rich heterogeneous networks that consist of multiple types of diversely interdependent entities \cite{Shi2017, sun2013mining}. 
This heterogeneity of real systems is often due to the fact that, in applications, data usually contains semantic information. 
For example in research publication networks, nodes can represent authors, papers, or venues and edges can represent coauthorships, references, or journal/conference appearances.
Such heterogeneous graph data can be represented by an arbitrary number of matrices and tensors that are coupled with respect to one or more types as shown in Figure~\ref{fig:coupled-matrix-tensor}.

Clusters in heterogeneous graphs that contain multiple types of nodes give rise to communities that are significantly more complex. 
Joint analysis of multiple graphs may capture fine-grained clusters that would not be captured by clustering each graph individually as shown in~\cite{banerjee2007multi,pcmf-dsaa}.
For instance, simultaneously clustering different types of entities/nodes in the heterogeneous graph based on multiple relations where each relation is represented as a matrix or a tensor (Figure~\ref{fig:coupled-matrix-tensor}).
It is due to this complexity and the importance of explicitly modeling how those entity types mix to form complex communities that make the problem of heterogeneous graph clustering a lot more challenging.
Moreover, the complexity, representation, and modeling of the heterogeneous graph data itself also makes this problem challenging (See Figure~\ref{fig:coupled-matrix-tensor}).
Extensions of clustering methods for homogeneous graphs to heterogeneous graphs are often nontrivial.
Many methods require complex schemas and are very specialized, allowing for two graphs with particular structure.
Furthermore, most clustering methods only consider first order structures in graphs, \ie, edge connectivity information.
However, higher-order structures play a non-negligible role in the organization of a network. 

Higher-order connectivity patterns such as small induced subgraphs called graphlets (network motifs) are known to be the fundamental building blocks of simple homogeneous networks~\cite{Milo2002} and are essential for modeling and understanding the fundamental components of these networks~\cite{pgd,pgd-kais,benson2016higher}.
However, such (untyped) graphlets are \emph{unable} to capture the rich (typed) connectivity patterns in more complex networks such as those that are heterogeneous, labeled, signed, or attributed.
In heterogeneous graphs (Figure~\ref{fig:coupled-matrix-tensor}), nodes and edges can be of different types and explicitly modeling such types is crucial.
In this work, we introduce the notion of a typed-graphlet and use it to uncover the higher-order organization of rich heterogeneous networks.
The notion of a typed-graphlet captures both the connectivity pattern of interest and the types.
We argue that typed-graphlets are the fundamental \emph{building blocks of heterogeneous networks}.
Note homogeneous, labeled, signed, and attributed graphs are all special cases of heterogeneous graphs as shown in Section~\ref{sec:framework}.

In this paper, we propose a general framework for \emph{higher-order clustering in heterogeneous graphs}. 
The framework explicitly incorporates heterogeneous higher-order information by counting typed graphlets that explicitly capture node and edge types.
Typed graphlets generalize the notion of graphlets to rich heterogeneous networks as they explicitly capture the higher-order typed connectivity patterns in such networks.
Using these as a basis, we propose the notion of \emph{typed-graphlet conductance} that generalizes the traditional conductance to higher-order structures in heterogeneous graphs.
The proposed approach reveals the higher-order organization and composition of rich heterogeneous complex networks. 
Given a graph and a typed-graphlet of interest $H$, the framework forms the weighted typed-graphlet adjacency matrix $\mW_{G^H}$ by counting the frequency that two nodes co-occur in an instance of the typed-graphlet. Next, the typed-graphlet Laplacian matrix is formed from $\mW_{G^H}$ and the eigenvector corresponding to the second smallest eigenvalue is computed. 
The components of the eigenvector provide an ordering $\sigma$ of the nodes that produce nested sets $S_k = \{\sigma_1, \sigma_2,\ldots, \sigma_k\}$ of increasing size $k$. 
We demonstrate theoretically that $S_k$ with the minimum typed-graphlet conductance is a near-optimal higher-order cluster.

The framework provides mathematical guarantees on the optimality of the higher-order clustering obtained.
The theoretical results extend to typed graphlets of arbitrary size and avoids restrictive special cases required in prior work.
Specifically, we prove a Cheeger-like inequality for \emph{typed-graphlet conductance}.
This gives bounds $OPT\le APPX\le C\sqrt{OPT}$ where $OPT$ is the minimum typed-graphlet conductance, $APPX$ is the value given by Algorithm \ref{algo:spectral}, and $C$ is a constant---at least as small as $\sqrt{OPT}$ which depends on the number of edges in the chosen typed graphlet.
Notably, the bounds of the method depend directly on the number of edges of the arbitrarily chosen typed graphlet (as opposed to the number of nodes) and inversely on the quality of connectivity of occurrences of the typed graphlet in a heterogeneous graph.
This is notable as the formulation for homogeneous graphs and untyped graphlets proposed in~\cite{benson2016higher} is in terms of nodes and requires different theory for untyped-graphlets with a different amount of nodes (\eg, untyped graphlets with 3 nodes vs. 4 nodes and so on).
In this work, we argue that it is not the number of nodes in a graphlet that are important, but the number of edges.
This leads to a more powerful, simpler, and general framework that can serve as a basis for analyzing higher-order spectral methods.
Furthermore, even in the case of untyped graphlets and homogeneous graphs, the formulation in this work leads to tighter bounds for certain untyped graphlets.
Consider a 4-node star and 3-node clique (triangle), both have 3 edges, and therefore would have the same bounds in our framework even though the number of nodes differ.
However, in~\cite{benson2016higher}, the bounds for the 4-node star would be different (and larger) than the 3-node clique.
This makes the proposed formulation and corresponding bounds more general and in the above case provides tighter bounds compared to~\cite{benson2016higher}.

The experiments demonstrate the effectiveness of the approach quantitatively for three important tasks.
First, we demonstrate the approach for revealing high quality clusters across a wide variety of graphs from different domains.
In all cases, it outperforms a number of state-of-the-art methods with an overall improvement of $43$x over all graphs and methods.
Second, we investigate the approach for link prediction.
In this task, we derive higher-order typed-graphlet node embeddings (as opposed to clustering) and use these embeddings to learn a predictive model.
Compared to state-of-the-art methods, the approach achieves an overall improvement in $F_1$ and AUC of 18.7\% and 14.4\%, respectively.
Finally, we also demonstrate the effectiveness of the approach quantitatively for graph compression where it is shown to achieve a mean improvement of $20.8\%$ across all graphs and methods.  
Notably, these application tasks all leverage different aspects of the proposed framework. For instance, link prediction uses the higher-order node embeddings given by our approach whereas graph compression leverages the proposed typed-graphlet spectral ordering (Definition~\ref{def:typed-graphlet-spectral-ordering}).

The paper is organized as follows. 
Section~\ref{sec:framework} describes the general framework for higher-order spectral clustering 
whereas Section~\ref{sec:theory} proves a number of important results including mathematical guarantees on the optimality of the higher-order clustering.
Next, Section \ref{sec:exp} demonstrate the effectiveness of the approach quantitatively for a variety of important applications including clustering, link prediction, and graph compression. 
Section \ref{sec:related-work} discusses and summarizes related work.
Finally, Section \ref{sec:conc} concludes.

\begin{figure}[t!]
\centering
\includegraphics[width=0.6\linewidth]{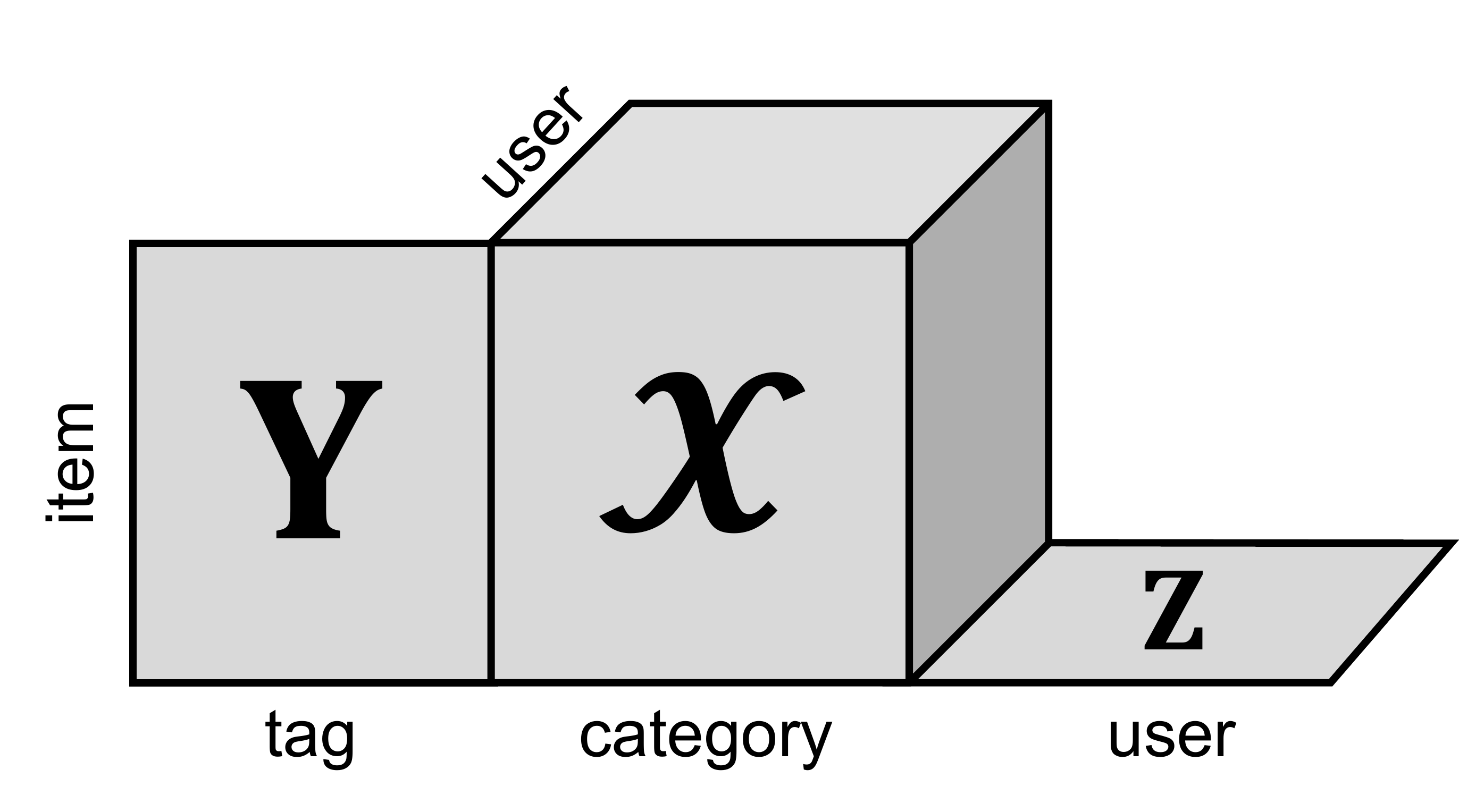}
\vspace{-2mm}
\caption{Heterogeneous graph represented as a third-order tensor and two matrices that all share at least one type. The third-order tensor $\tX$ can be coupled with the item by tag matrix $\mY$ and the social network (user by user) matrix $\mZ$.}
\label{fig:coupled-matrix-tensor}
\end{figure}

\section{Framework} \label{sec:framework}
In this work, we propose a general framework for higher-order clustering in heterogeneous graphs.
Table \ref{table:notation} lists all our notation.

\subsection{Heterogeneous Graph Model} \label{sec:heterogeneous-graph-model}
We represent a heterogeneous complex system using the following heterogeneous graph model.

\begin{Definition}[\sc Heterogeneous Graph] \label{def:hetero-graph}
A heterogeneous graph is an ordered tuple $G=(V,E,\psi,\xi)$ comprised of
\begin{enumerate}
\item a graph $(V,E)$ where $V$ is the node set and $E$ is the edge set,
\item a mapping $\psi:V\rightarrow\mathcal{T}_V$ referred to as the node-type mapping where $\mathcal{T}_V$ is a set of node types,
\item a mapping $\xi:E\rightarrow\mathcal{T}_E$ referred to as the edge-type mapping where $\mathcal{T}_E$ is a set of edge types.
\end{enumerate}
We denote the node set of a heterogeneous graph $G$ as $V(G)$ and its edge set as $E(G)$.
\end{Definition}\noindent

A homogeneous graph can be seen as a special case of a heterogeneous graph where $|\mathcal{T}_V|=|\mathcal{T}_E|=1$. 
Note that a heterogeneous graph may be unweighted or weighted and it may be undirected or directed, depending on the underlying graph structure.
Moreover, it may also be signed or labeled $Y = \{y_1, y_2, \ldots\}$ where $y_i$ corresponds to a label assigned to node $v_i$ (or edge $e_i$).

In general, a heterogeneous network can be represented by an arbitrary number of matrices and tensors that are coupled, \ie, the tensors and matrices share at least one type with each other~\cite{acar2011all,rossi16collective-factor}.
See Figure~\ref{fig:coupled-matrix-tensor} for an example of a heterogeneous network represented as a coupled matrix-tensor.

\subsection{Graphlets}\label{sec:graphlets}
Graphlets are small connected induced subgraphs~\cite{Przulj2004,pgd}. 
The simplest nontrivial graphlet is the 1st-order structure of a node pair connected by an edge. 
Higher-order graphlets correspond to graphlets with greater number of nodes and edges.
Most graph clustering algorithms only take into account edge connectivity, 1st-order graphlet structure, when determining clusters. 
Moreover, these methods are only applicable for homogeneous graphs.
For example, spectral clustering on the normalized Laplacian of the adjacency matrix of a graph partitions it in a way that attempts to minimize the amount of edges, 1st-order structures, cut \cite{Kannan2004}.

In this section, we introduce a more general notion of graphlet called \emph{typed-graphlet} that naturally extends to both homogeneous and heterogeneous networks.
In this paper, we will use $G$ to represent a graph and $H$ or $F$ to represent graphlets.

\subsubsection{Untyped graphlets} \label{sec:untyped-graphlets}
We begin by defining graphlets for homogeneous graphs with a single type.

\begin{Definition}[\sc Untyped Graphlet]\label{def:graphlet}
An untyped graphlet of a homogeneous graph $G$ is a connected, induced subgraph of $G$.
\end{Definition}

Given an untyped graphlet in some homogeneous graph, it may be the case that we can find other topologically identical ``appearances" of this structure in that graph. 
We call these appearances \emph{untyped-graphlet instances}.

\begin{Definition}[\sc Untyped-Graphlet Instance]\label{def:graphlet-instance}
An instance of an untyped graphlet $H$ in homogeneous graph $G$ is an untyped graphlet $F$ in $G$ that is isomorphic to $H$.
\end{Definition}

As we shall soon see, it will be important to refer to the set of all instances of a given graphlet in a graph. 
Forming this set is equivalent to determining the subgraphs of a graph isomorphic to the given graphlet. 
Nevertheless, we usually only consider graphlets with up to four or five nodes, and have fast methods for discovering instances of such graphlets~\cite{pgd, pgd-kais, ahmed16bigmine, ahmed16bigdata, rossi17graphlet-est}.

\subsubsection{Typed graphlets} \label{sec:typed-graphlets}
In heterogeneous graphs, nodes and edges can be of different types and so explicitly (and jointly) modeling such types is essential (Figure~\ref{fig:coupled-matrix-tensor}).
To generalize higher-order clustering to handle such networks, we introduce the notion of a typed-graphlet that explicitly captures both the connectivity pattern of interest and the types.
Notice that typed-graphlets are a generalization of untyped-graphlets and thus are a more powerful representation.

\begin{Definition}[\sc Typed Graphlet]\label{def:typed-graphlet}
A typed graphlet of a heterogeneous graph $G=(V,E,\psi,\xi)$ is a connected induced heterogeneous subgraph $H=(V',E',\psi',\xi')$ of $G$ in the following sense:
\begin{enumerate}
\item $(V',E')$ is an untyped graphlet of $(V,E)$,

\item $\psi'=\psi|_{V'}$, that is, $\psi'$ is the restriction of $\psi$ to $V'$

\item $\xi'=\xi|_{E'}$, that is, $\xi'$ is the restriction of $\xi$ to $E'$.
\end{enumerate}
\end{Definition}

We can consider the topologically identical ``appearances" of a typed graphlet in a graph that preserve the type structure.

\begin{Definition}[\sc Typed-Graphlet Instance]\label{def:typed-graphlet-instance}
An instance of a typed graphlet $H=(V',E',\psi',\xi')$ of heterogeneous graph $G$ is a typed graphlet $F=(V'',E'',\psi'',\xi'')$ of $G$ such that:
\begin{enumerate}
\item $(V'',E'')$ is isomorphic to $(V',E')$, 
\item $\mathcal{T}_{V''}=\mathcal{T}_{V'}$ and $\mathcal{T}_{E''}=\mathcal{T}_{E'}$, that is, the sets of node and edge types are correspondingly equal.
\end{enumerate}
The set of unique typed-graphlet instances of $H$ in $G$ is denoted as $I_G(H)$.
\end{Definition}

Note that we are not interested in preserving the type structure via the isomorphism, only its existence, that is, we are not imposing the condition that the node and edge types coincide via the graph isomorphism. This condition is too restrictive.

\subsubsection{Motifs} \label{sec:motifs}
Before we proceed, we briefly address some discrepancies between our definition of graphlets and that of papers such as \cite{benson2016higher,Arenas2005}. Although it might be a simple matter of semantics, the differences should be noted and clarified to avoid confusion. Some papers refer to what we refer to graphlets as \emph{motifs}. Yet, motifs usually refer to \emph{recurrent} and \emph{statistically significant} induced subgraphs ~\cite{Przulj2004, Milo2002}.

To find the motifs of a graph, one must compare the frequency of appearances of a graphlet in the graph to the expected frequency of appearances in an ensemble of random graphs in a null model associated to the underlying graph. 
Current techniques for computing the expected frequency in a null model requires us to generate a graph that follows the null distribution and then compute the graphlet frequencies in this sample graph \cite{Milo2002, Albert2002}. These tasks are computationally expensive for large networks as we have to sample many graphs from the null distribution. 
On the other hand, any graphlet can be arbitrarily specified in a graph and does not depend on being able to determine whether it is is statistically significant. In any case, a motif is a special type of graphlet, so we prefer to work with this more general object.

\begin{table}[t!]
\caption{Summary of notation. Matrices are bold, upright roman letters.}
\vspace{-2mm}
\scalebox{0.9}{
\centering 
\fontsize{8}{8.5}\selectfont
\setlength{\tabcolsep}{6pt} 
\label{table:notation}
\hspace*{-2.5mm}
\def\arraystretch{1.38}
\begin{tabularx}{1.10\linewidth}
{@{}r X@{}} 
\toprule

$G$ & graph \\ 
$V(G)$ & node set of $G$ \\
$E(G)$ & edge set of $G$ \\
$H,F$ & graphlet of $G$ \\
$I_G(H)$ & set of unique instances of $H$ in $G$ \\
$\mW_{G^H}$ & typed-graphlet adjacency matrix of $G$ based on $H$ \\
$\mL_{G^H}$ & typed-graphlet normalized Laplacian of $G$ based on $H$ \\
$G^H$ & weighted heterogeneous graph induced by $\mW_G^H$ \\
$S$ & subset of $V(G)$ \\
$(S,\bar{S})$ & cut of $G$ where $\bar{S}=V(G)\backslash S$ \\
$\deg_G(v)$ & degree of node $v\in V(G)$ \\
$\deg_G^H(v)$ & typed-graphlet degree of node $v\in V(G)$ based on $H$\\
$\vol_G(S)$ & volume of $S$ under $G$ \\
$\vol_G^H(S)$ & typed-graphlet volume of $S$ based on $H$ under $G$ \\
$\cut_G(S,\bar{S})$ & cut size of $(S,\bar{S})$ under $G$ \\
$\cut_G^H(S,\bar{S})$ & typed-graphlet cut size of $(S,\bar{S})$ based on $H$ under $G$ \\
$\phi_G(S,\bar{S})$ & conductance of $(S,\bar{S})$ under $G$ \\
$\phi_G^H(S,\bar{S})$ & typed-graphlet conductance of $(S,\bar{S})$ based on $H$ under $G$ \\
$\phi(G)$ & conductance of $G$ \\
$\phi^H(G)$ & typed-graphlet conductance of $G$ based on $H$ \\
\bottomrule
\end{tabularx}
}
\end{table}

\subsection{Typed-Graphlet Conductance} \label{sec:conductance}
In this section, we introduce the measure that will score the quality of a heterogeneous graph clustering built from typed graphlets. 
It is extended from the notion of conductance defined as:
\begin{align*}
\phi(S,\bar{S})=\frac{\cut(S,\bar{S})}{\min\left(\vol(S),\vol(\bar{S})\right)}
\end{align*}
where $(S,\bar{S})$ is a cut of a graph, $\cut(S,\bar{S})$ is the number of edges crossing cut $(S,\bar{S})$ and $\vol(S)$ is the total degrees of the vertices in cluster $S$ \cite{Fortunato2010, Kannan2004}. 
Note that its minimization achieves the sparsest balanced cut in terms of the total degree of a cluster.

The following definitions apply for a fixed heterogeneous graph and typed graphlet. 
Assume we have a heterogeneous graph $G$ and a typed graphlet $H$ of $G$.

\begin{note}
We denote the set of unique instances of $H$ in $G$ as $I_G(H)$.
\end{note}

\begin{Definition}[\sc Typed-Graphlet Degree]
The typed-graphlet degree based on $H$ of a node $v\in V(G)$ is the total number of incident edges to $v$ over all unique instances of $H$. We denote and compute this as $$\deg_G^H(v)=\sum_{F\in I_G(H)}\left|\{e\in E(F)\ |\ v\in e\}\right|.$$
\end{Definition}

\begin{Definition}[\sc Typed-Graphlet Volume]\label{def:typed-graphlet-volume}
The typed-graphlet volume based on $H$ of a subset of nodes $S\subset V(G)$ is the total number of incident edges to any node in $S$ over all instances of $H$. In other words, it is the sum of the typed-graphlet degrees based on $H$ over all nodes in $S$. We denote and compute this as
$$\vol_G^H(S)=\sum_{v\in S}\deg_G^H(v).$$
\end{Definition}

Recall that a cut in a graph $G$ is a partition of the underlying node set $V(G)$ into two proper, nonempty subsets $S$ and $\bar{S}$ where $\bar{S}=V(G)\backslash S$. We denote such a cut as an ordered pair $(S,\bar{S})$. For any given cut in a graph, we can define a notion of cut size.

\begin{Definition}[\sc Typed-Graphlet Cut Size]\label{def:typed-graphlet-cut-size}
The typed-graphlet cut size based on $H$ of a cut $(S,\bar{S})$ in $G$ is the number of unique instances of $H$ crossing the cut. We denote and compute this as $$\cut_G^{H}(S,\bar{S})=\left|\{F\in I_G(H)\ |\ V(F)\cap S\neq\varnothing, V(F)\cap\bar{S}\neq\varnothing\}\right|.$$
\end{Definition}

Note that a graphlet can cross a cut with any of its edges. 
Therefore, it has more ways in which it can add to the cut size than just an edge.

Having extended notions of volume and cut size for higher-order typed substructures, we can naturally introduce a corresponding notion of conductance.

\begin{Definition}[\sc Typed-Graphlet Conductance]\label{def:typed-graphlet-conductance}
The typed-graphlet conductance based on $H$ of a cut $(S,\bar{S})$ in $G$ is
$$\phi_G^H(S,\bar{S})=\frac{\cut_G^H(S,\bar{S})}{\min\big(\vol_G^H(S),\vol_G^H(\bar{S})\big)},$$
and the typed-graphlet conductance based on $H$ of $G$ is defined to be the minimum typed-graphlet conductance based on $H$ over all possible cuts in $G$:
\begin{equation} \label{eq:typed-graphlet-conductance}
\phi^H(G)=\min_{S\subset V(G)}\phi_G^H(S,\bar{S}).
\end{equation}
\end{Definition}\noindent

The cut which achieves the minimal typed-graphlet conductance corresponds to the cut that minimizes the amount of times instances of $H$ are cut and still achieves a balanced partition in terms of instances of $H$ in the clusters.

\subsection{Typed-Graphlet Laplacian}
In this section, we introduce a notion of a higher-order Laplacian of a graph. Assume we have a heterogeneous graph $G$ and a typed graphlet $H$ of $G$.

\subsubsection{Typed-graphlet adjacency matrix}
Suppose we have the set $I_G(H)$. 
Then, we can form a matrix that has the same dimensions as the adjacency matrix of $G$ and has its entries defined by the count of unique instances of $H$ containing edges in $G$. 

\begin{Definition}[\sc Typed-Graphlet Adjacency Matrix]\label{def:typed-graphlet-adjacency-matrix}
Suppose that $V(G)=\{v_1,\dots,v_n\}$. The typed-graphlet adjacency matrix $\mW_{G^H}$ of $G$ based on $H$ is a \emph{weighted matrix} defined by
$$(\mW_{G^H})_{ij}=\sum_{F\in I_G(H)}\mathbf{1}\big(\{v_i,v_j\}\in E(F)\big)$$
for $i,j=1,\dots,n$. That is, the $ij$-entry of \, $\mW_{G^H}$ is equal to the number of unique instances of $H$ that contain nodes $\{v_i,v_j\}\subset V(G)$ as an edge.
\end{Definition}

Having defined $\mW_{G^H}$, a weighted adjacency matrix on the set of nodes $V(G)$, we can induce a weighted graph. We refer to this graph as the graph induced by $\mW_{G^H}$ and denote it as $G^H$.

\begin{note}
From the definition of $\mW_{G^H}$, we can easily show that $E(F)\subset E(G^H)$ for any $F\in I_G(H)$.
\end{note}

\subsubsection{Typed-graphlet Laplacian}
We can construct the weighted normalized Laplacian of $W_{G^H}$:
$$\mL_{G^H}^{\ }=\mI-\mD_{G^H}^{-1/2}\mW_{G^H}^{\ }\mD_{G^H}^{-1/2}$$ where $\mD_{G^H}$ is defined by $$(\mD_{G^H})_{ii}=\sum_j(\mW_{G^H})_{ij}$$
for $i=1,\dots,n$. We also refer to this Laplacian as the \emph{typed-graphlet normalized Laplacian} based on $H$ of $G$. The normalized typed-graphlet Laplacian is the fundamental structure for the method we present in Section \ref{sec:typed-graphlet-spectral-clustering}.

\subsection{Typed-Graphlet Spectral Clustering} \label{sec:typed-graphlet-spectral-clustering}
In this section, we present an algorithm for approximating the optimal solution to the minimum typed-graphlet conductance optimization problem:
\begin{align}\label{eq:graphlet-conductance-opt}
S_\text{best}=\argmin_{S\subset V(G)}\phi_G^H(S,\bar{S})
\end{align}

Minimizing the typed-graphlet conductance encapsulates what we want: the solution achieves a bipartition of $G$ that minimizes the number of instances of $H$ that are cut and is balanced in terms of the total graphlet degree contribution of all instances of $H$ on each partition.

The issue is that minimizing typed-graphlet conductance is $NP$-hard. To see this, consider the case where your graphlet is the 1st-order graphlet, that is, a pair of nodes connected by an edge. Minimizing the standard notion of conductance, which is known to be NP-hard ~\cite{Cook1971}, reduces to minimizing this special case of 1st-order untyped-graphlet conductance minimization. Therefore, obtaining the best graphlet-preserving clustering for large graphs is an intractible problem. We can only hope to achieve a near-optimal approximation.

\subsubsection{Algorithm}
We present a typed-graphlet spectral clustering algorithm for finding a provably near-optimal bipartition in Algorithm~\ref{algo:spectral}. 
We build a sweeping cluster in a greedy manner according to the typed-graphlet spectral ordering defined as follows.

\begin{Definition}[\sc Typed-Graphlet Spectral Ordering] \label{def:typed-graphlet-spectral-ordering}
Let $v$ denote the eigenvector corresponding to the $2$nd smallest eigenvalue of the normalized typed-graphlet Laplacian $\mL_{G^H}$. The typed-graphlet spectral ordering is the permutation
\begin{equation*}
\sigma = (i_1, i_2, \dots, i_n)
\end{equation*}
of coordinate indices $(1,2,\dots, n)$ such that
\begin{equation*}
v_{i_1} \le v_{i_2}\le \dots\le v_{i_n},
\end{equation*}
that is, $\sigma$ is the permutation of coordinate indices of $v$ that sorts the corresponding coordinate values from smallest to largest.
\end{Definition}\noindent

\begin{algorithm} 
\DontPrintSemicolon
\caption{Typed-Graphlet Spectral Clustering}
\SetAlgoLined
\KwIn{Heterogeneous graph $G$, typed graphlet $H$} 
\KwOut{Near-optimal cluster}
$\mW_{G^H} \gets \text{typed-graphlet adjacency matrix of } G \text{ based on } H$

$N\gets \text{number of connected components of } G^H$

$\phi_{\min} \gets \infty$

$S_{\text{best}}\gets \text{initialize space for best cluster}$

\For{$i\gets 1\textbf{ to } N$}{
$\mW \gets \text{submatrix of }\mW_{G^H}\text{ on connected component } i$

$\mL\gets\text{typed-graphlet normalized Laplacian of }W$

$\vv_2 \gets \text{eigenvector of }\mL\text{ with 2nd smallest eigenvalue}$

$\sigma \gets \text{argsort}(\vv_2)$

$\phi \gets \min_k\phi_{G^H}(S_k,\bar{S}_k),\text{ where }S_k=\{\sigma_1,\dots,\sigma_k\}$

\If{$\phi < \phi_{\min}$}{
$\phi_{\min}\gets \phi$

$S\gets \argmin_k\phi_{G^H}(S_k,\bar{S}_k)$

\eIf{$|S|<|\bar{S}|$}{
$S_\text{best}\gets S$
}{
$S_\text{best}\gets\bar{S}$
}   
}
}
\Return{$S_{\text{best}}$}
\label{algo:spectral}
\end{algorithm}

\subsubsection{Extensions}
Algorithm \ref{algo:spectral} generalizes the spectral clustering method for standard conductance minimization~\cite{shi2000normalized} and untyped-graphlet conductance minimization.
We demonstrated the reduction of standard conductance minimization above. 
Untyped-graphlet conductance minimization is also generalized since homogeneous graphs can be seen as heterogeneous graphs with a single node and edge type.
It is straightforward to adapt the framework to other arbitrary (sparse) cut functions such as ratio cuts~\cite{gaertler2005clustering}, normalized cuts~\cite{shi2000normalized}, 
bisectors~\cite{gaertler2005clustering}, normalized association cuts~\cite{shi2000normalized}, among others~\cite{Schaeffer2007,Fortunato2010,shi2000normalized}.

Multiple clusters can be found through simple recursive bipartitioning~\cite{Kannan2004}.
We could also embed the lower $k$ eigenvectors of the normalized typed-graphlet Laplacian into a lower dimensional Euclidean space and perform $k$-means, or any other Euclidean clustering algorithm, then associate to each node its corresponding cluster in this space \cite{Ng2002, Kannan2004}. 
It is also straightforward to use multiple typed-graphlets for clustering or embeddings as opposed to using only a single typed-graphlet independently.
For instance, the higher-order typed-graphlet adjacency matrices can be combined in some fashion (\eg, summation) and may even be assigned weights based on the importance of the typed-graphlet.
Moreover, the typed-graphlet conductance can be adapted in a straightforward fashion to handle multiple typed-graphlets.

\subsubsection{Discussion}
Benson et al.~\cite{benson2016higher} refers to their higher-order balanced cut measure as \emph{motif conductance} and it differs from our proposed notion of typed-graphlet conductance. 
However, the definition used matches more with a generalization known as the \emph{edge expansion}. 
The edge expansion of a cut $(S,\bar{S})$ is defined as
\begin{align}
\psi(S,\bar{S})=\frac{\cut(S,\bar{S})}{\min(|S|,|\bar{S}|)}.
\end{align}
The balancing is in terms of the number of vertices in a cluster. Motif conductance was defined with a balancing in terms of the number of vertices in any graphlet instance. To be precise, for any set of vertices $S$, let the cluster size of $S$ in $G$ based on $H$ be
\begin{align}
|S|_G^H&=\sum_{F\in I_G(H)}\sum_{v\in V(F)}\mathbf{1}(v\in S)=\sum_{v\in S}\sum_{F\in I_G(H)}\mathbf{1}\left(v\in V(F)\right).
\end{align}
Note that this does not take into account the degree contributions of each graphlet, only its node count contributions to a cluster $S$. 
In terms of our notation, untyped ``motif conductance" of a cut $(S,\bar{S})$ is defined in that work as
\begin{align*}
\psi_G^H(S,\bar{S})=\frac{\cut_G^H(S,\bar{S})}{\min\left(|S|_G^H,|\bar{S}|_G^H\right)}.
\end{align*}
Since this does not take into account node degree information, this is more of a generalization of edge expansion \cite{hoory2006expander, alon1997edge}, ``graphlet expansion", if you will, rather than conductance.
The difference is worth noting because it has been shown that conductance minimization gives better partitions than expansion minimization \cite{Kannan2004}. By only counting nodes, we give equal importance to all the vertices in a graphlet. Arguably, it is more reasonable to give greater importance to the vertices that not only participate in many graphlets but also have many neighbors within a graphlet and give lesser importance to vertices that have more neighbors that do not participate in a graphlet or do not have many neighbors within a graphlet. Our definition of typed-graphlet volume captures this idea to give an appropriate general notion of conductance.

\subsection{Typed-Graphlet Node Embeddings} \label{sec:embeddings}
Algorithm~\ref{alg:higher-order-typed-graphlet-embeddings} summarizes the method for deriving higher-order typed motif-based node embeddings (as opposed to clusters/partitions of nodes, or an ordering for compression/analysis, see Section~\ref{sec:exp-graph-compression}).
In particular, given a typed-graphlet adjacency matrix, Algorithm~\ref{alg:higher-order-typed-graphlet-embeddings} outputs a $N \times D$ matrix $\mZ$ of node embeddings.
For graphs with many connected components, Algorithm~\ref{alg:higher-order-typed-graphlet-embeddings} is called for each connected component of $G^{H}$ and the resulting embeddings are stored in the appropriate locations in the overall embedding matrix $\mZ$.

Multiple typed-graphlets can also be used to derive node embeddings.
One approach that follows from~\cite{HONE} is to derive low-dimensional node embeddings for each typed-graphlet of interest using Algorithm~\ref{alg:higher-order-typed-graphlet-embeddings}. 
After obtaining all the node embeddings for each typed-graphlet, we can simply concatenate them all into one single matrix $\mY$.
Given $\mY$, we can simply compute another low-dimensional embedding to obtain the final node embeddings that capture the important latent features from the node embeddings from different typed-graphlets.

\begin{algorithm} 
\DontPrintSemicolon
\caption{Typed-Graphlet Spectral Embedding}
\label{alg:higher-order-typed-graphlet-embeddings}
\SetAlgoLined
\LinesNumbered
\KwIn{Heterogeneous graph $G$, typed graphlet $H$, embedding dimension $D$}
\KwOut{Higher-order embedding matrix $\mZ \in \RR^{N \times D}$ for $H$}
\medskip
$\!\big(\mW_{G^H}\big)_{ij} \!\leftarrow \# \text{ instances of H containing } i \text{ and } j, \; \forall (i,j) \in E$

$\mD_{G^H} \leftarrow \!\text{typed-graphlet degree matrix } \big(\mD_{G^H}\big)_{ii} \!= \!\sum_{j} \!\big(\mW_{G^H}\big)_{ij}$

$\vx_1, \vx_2, \ldots, \vx_D \leftarrow $ eigenvectors of $D$ smallest eigenvalues of $\quad\quad\quad \mL_{G^{H}} = \eye - \mD_{G^H}^{-1/2}\mW_{G^H}\mD_{G^H}^{-1/2}$

$Z_{ij} \leftarrow X_{ij} \Big/ \sqrt{\sum_{j=1}^{D} X_{ij}^2}$

\Return{$\mZ = \big[\, \vz_1\;\,\, \vz_2\;\, \cdots \;\, \vz_n \,\big]^T \!\in \RR^{N \times D}$}
\end{algorithm}

\section{Theoretical Analysis}\label{sec:theory}
In this section, we show the near-optimality of Algorithm~\ref{algo:spectral}. 
The idea is to translate what we know about ordinary conductance for weighted homogeneous graphs, for which there has been substantial theory developed \cite{chung1997spectralbook, chung2000weighted, chung1996laplacians}, to this new measure we introduce of typed-graphlet conductance by relating these two quantities. Through this association, we can derive Cheeger-like results for $\phi^H(G)$ and for the approximation given by the typed-graphlet spectral clustering algorithm (Algorithm~\ref{algo:spectral}).
As in the previous section, assume we have a heterogeneous graph $G$ and a typed graphlet $H$. Also, assume we have the weighted graph $G^H$ induced from the typed-graphlet adjacency matrix $\mW_G^H$. 

We prove two lemmas from which our main theorem will immediately hold. Lemma \ref{lem:1} shows that the typed-graphlet volume and ordinary volume measures match: total typed-graphlet degree contributions of typed-graphlet instances matches with total counts of typed-graphlet instances on edges for any given subset of nodes.
In contrast, Lemma~\ref{lem:2} shows that equality does not hold for the notions of cut size. The reason lies in the fact that for any typed-graphlet instance, typed-graphlet cut size on $G$ only counts the number of typed-graphlet instances cut whereas ordinary cut size on $G^H$ counts the number of times typed-graphlet instances are cut. Therefore, these two measure at least match and at most differ by a factor equal to the size of the $H$, which is a fixed value that is small for the typed graphlets we are interested in, of size 3 or 4. Thus, we are able to reasonably bound the discrepancy between the notions of cut sizes.

Using these two lemmas, we immediately get our main result in Theorem~\ref{thm:1} which shows the relationship between $\phi^H(G)$ and $\phi(G^H)$ in the form of tightly bound inequality that is dependent only on the number of edges in $H$.
From this theorem, we arrive at two important corollaries. In Corollary~\ref{cor:1}, we prove Cheeger-like bounds for typed-graphlet conductance. In Corollary \ref{cor:2}, we show that the output of Algorithm \ref{algo:spectral} gives a near-optimal solution up to a square root factor and it goes further to show bounds in terms of the optimal value $\phi^H(G)$ to show the constant of the approximation algorithm which depends on the second smallest eigenvalue of the typed-graphlet adjacency matrix and the number of edges in $H$. This last result does not give a purely constant-factor approximation to the graph conductance because of their dependence on $G$ and $H$, yet it still gives a very efficient, and non-trivial, approximation for fixed $G$ and $H$. 
Moreover, the second part of Corollary \ref{cor:2} provides intuition as to what makes a specific typed-graphlet a suitable choice for higher-order clustering. Typed-graphlets that have a good balance of small edge set size and strong connectivity in the heterogeneous graph---in the sense that the second eigenvalue of the normalized typed-graphlet Laplacian is large---will have a tighter upper bound to their approximation for minimum typed-graphlet conductance. Therefore, this last result in Corollary \ref{cor:2} provides a way to quickly and quantitatively measure how good a typed graphlet is for determining higher-order organization before even executing the clustering algorithm.

\begin{note}
In the case of the simple 1st-order untyped graphlet, \ie, a node pair with an interconnecting edge, we recover the results for traditional spectral clustering since $|E(H)|=1$ in this case. 
Furthermore, if $G$ is a homogeneous graph, \ie, $|\mathcal{T}_V|=|\mathcal{T}_E|=1$, we get the special case of untyped graphlet-based spectral clustering.  
Therefore, our framework generalizes the methods of traditional spectral clustering and untyped-graphlet spectral clustering for homogeneous graphs.
\end{note}

In the following analysis, we let $\mathbf{1}(\cdot)$ represent the Boolean predicate function and let $(\mW_{G^H})_e$ be the edge weight of edge $e$ in $G^H$.

\begin{lemma}\label{lem:1}
Let $S$ be a subset of nodes in $V(G)$. Then,
$$\vol_{G^H}(S)=\vol_G^H(S).$$
\end{lemma}
\begin{proof}
\begin{align}
\vol_{G^H}(S)&=\sum_{v\in S}\deg_{G^H}(v) \\
&=\sum_{v\in S}\sum_{e\in E(G^H)}\mathbf{1}(v\in e)\cdot(\mW_{G^H})_e \\
&=\sum_{v\in S}\sum_{e\in E(G^H)}\mathbf{1}(v\in e)\cdot\sum_{F\in I_G(H)}\mathbf{1}(e\in E(F)) \\
&=\sum_{v\in S}\sum_{F\in I_G(H)}\sum_{e\in E(G^H)}\mathbf{1}(v\in e)\cdot\mathbf{1}(e\in E(F)) \\
&=\sum_{v\in S}\sum_{F\in I_G(H)}\sum_{e\in E(F)}\mathbf{1}(v\in e) \\
&=\sum_{v\in S}\sum_{F\in I_G(H)}|\{e\in E(F)\ |\ v\in e\}| \\
&=\sum_{v\in S}\deg_G^H(v) \\
&=\vol_G^H(S).
\end{align}
\end{proof}

\renewcommand{\S}{\bar{S}}
\begin{lemma}\label{lem:2}
Let $(S,\bar{S})$ be a cut in $G$. Then,
$$\frac{1}{|E(H)|}\cut_{G^H}(S,\bar{S})\le\cut_G^H(S,\bar{S})\le\cut_{G^H}(S,\bar{S}).$$
\end{lemma}
\begin{proof}
For subsequent simplification, we define $[S,\S]$ to be the set of edges in $E(G^H)$ that cross cut $(S,\S)$:
\begin{align}
[S,\S]:=\{e\in E(G^H)\ |\ e\cap S\neq\varnothing, e\cap\S\neq\varnothing\}.
\end{align}
Then,
\begin{align}
\cut_{G^H}(S,\S)&=\sum_{e\in E(G^H)}\mathbf{1}\left(e\in[S,\S]\right)\cdot(\mW_{G^H})_e \\
&=\sum_{e\in E(G^H)}\mathbf{1}\left(e\in[S,\S]\right)\cdot\sum_{F\in I_G(H)}\mathbf{1}(e\in E(F)) \\
&=\sum_{F\in I_G(H)}\sum_{e\in E(G^H)}\mathbf{1}\left(e\in[S,\S]\right)\cdot\mathbf{1}(e\in E(F)) \\
&=\sum_{F\in I_G(H)}\sum_{e\in E(F)}\mathbf{1}\left(e\in[S,\S]\right) \\
&=\sum_{F\in I_G(H)}|E(F)\cap[S,\S]| \label{eq:1}
\end{align}

\noindent Note that for an instance $F\in I_G(H)$ such that $E(F)\cap[S,\S]\neq\varnothing$, there exists at least one edge in $E(F)$ cut by $(S,\S)$ and at most all edges in $E(F)$ are cut by $(S,\S)$. Clearly, if $E(F)\cap[S,\S]=\varnothing$, then no edge is cut by $(S,\S)$. This shows that for such an instance we have
\begin{align}
1\le|E(F)\cap[S,\S]|\le|E(F)|=|E(H)|.
\end{align}
Therefore, Equation \ref{eq:1} satisfies the following inequalities:
\begin{align}
\sum_{F\in I_G(H)}\mathbf{1}\left(E(F)\cap[S,\S]\neq\varnothing\right)&\le\cut_{G^H}(S,\S) \label{eq:2} \\
|E(H)|\cdot\sum_{F\in I_G(H)}\mathbf{1}\left(E(F)\cap[S,\S]\neq\varnothing\right)&\ge \cut_{G^H}(S,\S). \label{eq:3}
\end{align}
Referring to Definition \ref{def:typed-graphlet-cut-size} for typed-graphlet cut size and noting that since $H$ is a connected graph,
\begin{align}
\mathbf{1}(E(F)\cap[S,\S]\neq\varnothing)=\mathbf{1}(V(F)\cap S\neq\varnothing, V(F)\cap \S\neq\varnothing),
\end{align}
we find that
\begin{align}
\sum_{F\in I_G(H)}\mathbf{1}(E(F)\cap[S,\S])=\cut_G^H(S,\S).
\end{align}
Plugging this into Inequalities \ref{eq:2}-\ref{eq:3}, we get
\begin{align}
\cut_G^H(S,\S)\le\cut_{G^H}(S,\S)\le|E(H)|\cut_G^H(S,\S)
\end{align}
or, equivalently,
\begin{align}
\frac{1}{|E(H)|}\cut_{G^H}(S,\S)\le\cut_G^H(S,\S)\le\cut_{G^H}(S,\S)
\end{align}
\end{proof}

\begin{theorem}\label{thm:1}
\begin{align*}
\frac{1}{|E(H)|}\cdot\phi(G^H)\le\phi^H(G)\le\phi(G^H)
\end{align*}
\end{theorem}
\begin{proof}
Let $(S,\S)$ be any cut in $G$. From Lemma 2, we have that
\begin{align}\label{eq:thm1-1}
\frac{1}{|E(H)|}\cut_{G^H}(S,\S)\le\cut_G^H(S,\S)\le\cut_{G^H}(S,\S).
\end{align}
Lemma \ref{lem:1} shows that $\vol_{G^H}(S)=\vol_G^H(S)$. Therefore, if we divide these inequalities above by $\vol_{G^H}(S)=\vol_G^H(S)$, we get that
\begin{align}\label{eq:thm1-2}
\frac{1}{|E(H)|}\phi_{G^H}(S,\S)\le\phi_G^H(S,\S)\le\phi_{G^H}(S,\S)
\end{align}
by the definitions of conductance and typed-graphlet conductance. Since this result holds for any subset $S\subset V(G)$, it implies that
\begin{align}
\frac{1}{|E(H)|}\cdot\phi(G^H)\le\phi^H(G)\le\phi(G^H). \label{eq:4}
\end{align}
\end{proof}

\begin{corollary}\label{cor:1}
Let $\lambda_2$ be the second smallest eigenvalue of $\mL_{G^H}$. Then,
$$\frac{\lambda_2}{2|E(H)|}\le\phi^H(G)\le\sqrt{2\lambda_2}.$$
\end{corollary}
\begin{proof}
Cheeger's inequality for weighted undirected graphs (see proof in ~\cite{chung1997spectralbook}) gives
\begin{align}
\frac{\lambda_2}{2}\le\phi(G^H)\le\sqrt{2\lambda_2}.
\end{align}
Using these bounds for $\phi(G^H)$ and applying them to Theorem \ref{thm:1}, we find that
\begin{align}
\frac{\lambda_2}{2|E(H)|}\le\phi^H(G)\le\sqrt{2\lambda_2}.
\end{align}
\end{proof}

\begin{corollary}\label{cor:2}
Let $S$ be the cluster output of Algorithm \ref{algo:spectral} and let $\alpha=\phi_G^H(S,\bar{S})$ be its corresponding typed-graphlet conductance on $G$ based on $H$. Then,
\begin{align*}
\phi^H(G)\le\alpha\le\sqrt{4|E(H)|\phi^H(G)}.
\end{align*}
Moreover, if we let $\lambda_2$ be the second smallest eigenvalue of $\mL_{G^H}$, then
\begin{align*}
\phi^H(G)\le\alpha\le\beta\cdot\phi^H(G),
\end{align*}
where $$\beta=\sqrt{\frac{8}{\lambda_2}}\cdot|E(H)|,$$ showing that, for a fixed $G$ and $H$, Algorithm \ref{algo:spectral} is a $\beta$-approximation algorithm to the typed-graphlet conductance minimization problem.
\end{corollary}
\begin{proof}
Clearly $\phi^H(G)\le\alpha$ since $\phi^H(G)$ is the minimal typed-graphlet conductance. To prove the upper bound, let $(T,\bar{T})$ be the cut that achieves the minimal conductance on $G^H$, that is, $\phi_{G^H}(T,\bar{T})=\phi(G^H)$. Then,
\begin{align}
\alpha&\le\phi_G^H(T,\bar{T}) \label{eq:cor2-1}\\
&\le\phi_{G^H}(T,\bar{T}) \label{eq:cor2-2}\\
&\le\sqrt{2\lambda_2} \label{eq:cor2-3}\\
&\le\sqrt{4|E(H)|\phi^H(G)}.  \label{eq:cor2-4}
\end{align}
Inequality \ref{eq:cor2-1} follows from the fact that $\alpha$ achieves the minimal typed-graphlet conductance. Inequality \ref{eq:cor2-2} follows from Inequality \ref{eq:thm1-2} in Theorem \ref{thm:1}. 
Inequality \ref{eq:cor2-3} follows from Cheeger's inequality for weighted graphs (see \cite{chung1996laplacians} for a proof). 
Inequality \ref{eq:cor2-4} follows from the lower bound in Corollary \ref{cor:1}. \\

\noindent We can go a bit further to express the bounds entirely in terms of $\phi^H(G)$ by noting that
\begin{align}
\alpha&\le\sqrt{4|E(H)|\phi^H(G)} \\
&=\sqrt{\frac{4|E(H)|}{\phi^H(G)}}\cdot\phi^H(G) \\
&\le\sqrt{\frac{8}{\lambda_2}}\cdot|E(H)|\cdot\phi^H(G) \label{eq:cor2-5}
\end{align}
where Inequality \ref{eq:cor2-5} follows from the fact that $\sqrt{\phi^H(G)}\ge\frac{\lambda_2}{2|E(H)|}$ by the lower bound of Corollary \ref{cor:1}.
\newline
\end{proof}
\smallskip

\section{Experiments}\label{sec:exp}
This section empirically investigates the effectiveness of the proposed approach quantitatively for typed-graphlet spectral clustering (Section~\ref{sec:exp-quant-clusters}), link prediction using the higher-order node embeddings from our approach (Section~\ref{sec:exp-link-pred}) and the typed-graphlet spectral ordering for graph compression (Section~\ref{sec:exp-graph-compression}).
Unless otherwise mentioned, we use all 3 and 4-node graphlets.

\newcommand{\GraphletFigScale}{0.05}
\begin{table}[h!]
\centering
\renewcommand{\arraystretch}{1.2} 
\setlength{\tabcolsep}{3pt}
\caption{Network properties and statistics. 
Note $|\mathcal{T}_V|$ = \# of node types. 
Comparing the number of unique typed motifs that occur for each induced subgraph (\eg, there are 3 different typed 3-path graphlets that appear in yahoo).
}
\label{table:network-stats}
\vspace{-3mm}
\fontsize{7.5}{8.5}\selectfont
\begin{tabularx}{1.0\linewidth}{@{}l XXHH cH cc cccccc @{}HHHH H HHHHH
}
\toprule
\textbf{Graph}  &   
$|V|$  &  $|E|$  & & & $|\mathcal{T}_V|$ & &
\includegraphics[scale=0.6]{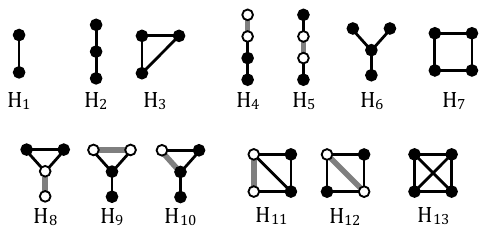} &
\includegraphics[scale=0.6]{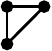} &
\includegraphics[scale=0.11]{fig4.pdf} &
\includegraphics[scale=0.6]{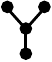} &
\includegraphics[scale=0.6]{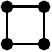} &
\includegraphics[scale=0.11]{fig7.pdf} &
\includegraphics[scale=0.10]{fig8.pdf} &
\includegraphics[scale=0.6]{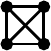} &&
\\
\midrule

\textsf{yahoo-msg} & 100.1k & 739.8k & && 2 & 1 & 
3 & 2 & 3 & 4 & 3 & 3 & 3 & 2 & 
& & & & \\

\textsf{dbpedia} & 495.9k & 921.7k & && 4 & 1 &
8 & 0 & 6 & 10 & 5 & 0 & 0 & 0 & 
& & & & \\

\textsf{digg} & 283.2k & 4.6M & && 2 & 1 & 
4 & 3 & 4 & 5 & 4 & 4 & 4 & 2 & 
& & & & \\

\textsf{movielens} & 28.1k & 170.4k & && 3 & 1 & 
7 & 1 & 6 & 9 & 6 & 3 & 3 & 0 & 
& & & & \\

\textsf{citeulike} & 907.8k & 1.4M & && 3 & 1 &
5 & 0 & 3 & 6 & 3 & 0 & 0 & 0 & \\

\textsf{fb-CMU} & 6.6k & 250k & && 3 & 1 & 
10 & 10 & 15 & 15 & 15 & 15 & 15 & 15 & 
& & & & \\

\textsf{reality} & 6.8k & 7.7k & && 2 & 1 & 
4 & 3 & 4 & 5 & 4 & 4 & 4 & 2 & 
& & & & \\

\textsf{gene} & 1.1k & 1.7k & && 2 & 1 & 
4 & 4 & 5 & 5 & 5 & 5 & 5 & 5 & \\

\textsf{citeseer} & 3.3k & 4.5k & && 6 & 1 & 
56 & 40 & 124 & 119 & 66 & 98 & 56 & 19 & \\

\textsf{cora} & 2.7k & 5.3k & && 7 & 1 & 
82 & 49 & 202 & 190 & 76 & 157 & 73 & 19 & \\

\textsf{webkb} & 262 & 459 & && 5 & 1 & 
31 & 21 & 59 & 59 & 23 & 51 & 32 & 8 & \\

\textsf{pol-retweet} & 18.5k & 48.1k & && 2 & 1 & 
4 & 4 & 5 & 5 & 5 & 5 & 5 & 4 & \\

\textsf{web-spam} & 9.1k & 465k & && 3 & 1 & 
10 & 10 & 15 & 15 & 15 & 15 & 15 & 15 & \\

\textsf{fb-relationship} & 7.3k & 44.9k & && 6 & 1 & 
50 & 47 & 112 & 109 & 85 & 106 & 89 & 77 & \\

\textsf{Enzymes-g123} & 90 & 127 & && 2 & 1 & 
4 & 3 & 5 & 5 & 5 & 4 & 3 & 0 & 
& & & & \\

\textsf{Enzymes-g279} & 60 & 107 & && 2 & 1 & 
4 & 4 & 5 & 5 & 5 & 5 & 5 & 0 & 
& & & & \\

\textsf{Enzymes-g293} & 96 & 109 & && 2 & 1 & 
4 & 1 & 5 & 5 & 1 & 2 & 1 & 0 & 
& & & & \\

\textsf{Enzymes-g296} & 125 & 141 & && 2 & 1 & 
4 & 1 & 4 & 5 & 2 & 1 & 1 & 0 & 
& & & & \\

\textsf{NCI109-g4008} & 90 & 105 & && 2 & 1 & 
3 & 0 & 3 & 3 & 0 & 0 & 0 & 0 & 
& & & & \\

\textsf{NCI109-g1709} & 102 & 106 & && 3 & 1 & 
5 & 0 & 5 & 5 & 1 & 0 & 0 & 0 & 
& & & & \\

\textsf{NCI109-g3713} & 111 & 119 & && 3 & 1 & 
4 & 0 & 6 & 4 & 0 & 0 & 0 & 0 & 
& & & & \\

\textsf{NCI1-g3700} & 111 & 119 & && 3 & 1 & 
4 & 0 & 6 & 4 & 0 & 0 & 0 & 0 & 
& & & & \\

\bottomrule
\end{tabularx}
\end{table}

\begin{table}[t!]
\centering
\setlength{\tabcolsep}{3.5pt}
\renewcommand{\arraystretch}{1.2} 
\scriptsize
\caption{Quantitative evaluation of the methods (external conductance~\protect\cite{almeida2011there}).
Note TGS is the approach proposed in this work.
The best result for each graph is bold.}
\vspace{-3mm}
\label{table:quant-conductance}
\small
\fontsize{8.0}{9.0}\selectfont
\begin{tabularx}{1.0\linewidth}{@{} lH XX XXXXX HHHH}

\toprule
&& 
\multicolumn{1}{@{}l}{\rotatebox{65}{\textbf{DS-H}}} &
\multicolumn{1}{l}{\rotatebox{65}{\textbf{KCore-H}}} &
\multicolumn{1}{l}{\rotatebox{65}{\textbf{LP-H}}} &
\multicolumn{1}{l}{\rotatebox{65}{\textbf{Louv-H}}}&
\multicolumn{1}{l}{\rotatebox{65}{\textbf{Spec-H}}}    &
\multicolumn{1}{l}{\rotatebox{65}{\textbf{GSpec-H}}} &
\multicolumn{1}{l}{\rotatebox{65}{\textbf{TGS}}} &
\\

\midrule
\textsf{yahoo-msg} && 0.5697 & 0.6624 & 0.2339 & 0.3288 & 0.0716 & 0.2000 & \textbf{0.0588} & \\
\textsf{dbpedia} && 0.7414 & 0.5586 & 0.4502 & 0.8252 & 0.9714 & 0.9404 & \textbf{0.0249} & \\ 
\textsf{digg} &&  0.4122 & 0.4443 & 0.7555 & 0.3232 & 0.0006 & \textbf{0.0004} & \textbf{0.0004} & \\ 
\textsf{movielens} && 0.9048 & 0.9659 & 0.7681 & 0.8620 & 0.9999 & 0.6009 & \textbf{0.5000} & \\ 
\textsf{citeulike} &&  0.9898 & 0.9963 & 0.9620 & 0.8634 & 0.9982 & 0.9969 & \textbf{0.7159} & \\ 
\textsf{fb-CMU} && 0.6738 & 0.9546 & 0.9905 & 0.8761 & 0.5724 & 0.8571 & \textbf{0.5000} & \\
\textsf{reality} &&  0.7619 & 0.3135 & 0.2322 & 0.1594 & 0.6027 & 0.0164 & \textbf{0.0080} & \\
\textsf{gene} && 0.8108 & 0.9298 & 0.9151 & 0.8342 & 0.4201 & 0.1667 & \textbf{0.1429} & \\ 
\textsf{citeseer} && 0.5000 & 0.6667 & 0.6800 & 0.6220 & 0.0526 & 0.0526 & \textbf{0.0333} & \\ 
\textsf{cora} && 0.0800 & 0.9057 & 0.8611 & 0.8178 & 0.0870 & 0.0870 & \textbf{0.0500} & \\
\textsf{webkb} && 0.2222 & 0.9286 &  0.6154 & 0.8646 & 0.6667 & 0.3333 & \textbf{0.2222} & \\ 
\textsf{pol-retweet} && 0.5686 & 0.6492 & 0.0291 & 0.0918 & 0.6676 & 0.0421 & \textbf{0.0220} & \\ 
\textsf{web-spam} && 0.8551 & 0.9331 & 0.9844 & 0.7382 & 0.9918 & 0.5312 & \textbf{0.5015} & \\ 
\textsf{fb-relationship} && 0.6249 & 0.9948 & 0.5390 & 0.8392 & 0.9999 & 0.5866 & \textbf{0.4972} & \\ 
\textsf{Enzymes-123} && 0.8667 & 0.8889 &  0.5696 & 0.6364 & 0.6768 & 0.5204 & \textbf{0.3902} & \\ 
\textsf{Enzymes-279} && 0.9999 & 0.4444 &  0.5179 & 0.4444 & 0.2929 & 0.3298 & \textbf{0.2747} & \\ 
\textsf{Enzymes-293} && 1.0000 & 0.4857 &  0.9444 & 0.3793 & 0.7677 & 0.5000 & \textbf{0.3023} & \\ 
\textsf{Enzymes-296} && 1.0000 & 0.7073 & 0.9286 & 0.7344 & 0.6406 & 0.5000 & \textbf{0.3212} & \\ 
\textsf{NCI109-4008} && 0.7619 & 0.4324 & 0.8462 & 0.8235 & 0.3500 & 0.4556 & \textbf{0.3204} & \\ 
\textsf{NCI109-1709} && 0.4000 & 0.3171 & 0.1429 & 0.4615 & 0.3922 & 0.3654 & \textbf{0.1333} & \\ 
\textsf{NCI109-3713} && 0.4074 & 0.3793 & 0.7500 & 0.4583 & 0.6667 & 1.0000 & \textbf{0.2000} & \\ 
\textsf{NCI1-3700} && 0.4074 & 0.3793 & 0.7500 & 0.4583 & 0.3333 & 0.6667 & \textbf{0.2500} & \\ 
\midrule
\textsc{Avg. Rank} && 
\multicolumn{1}{c}{\textcolor{theblue}{\bf 4.59}}  &  
\textcolor{theblue}{\bf 4.77}  &  
\textcolor{theblue}{\bf 4.64}  &  
\textcolor{theblue}{\bf 4.32}  &  
\textcolor{theblue}{\bf 4.27}  &  
\textcolor{theblue}{\bf 3.27}  &  
\multicolumn{1}{c}{\textcolor{theblue}{\bf 1}}  &  
\\
\bottomrule
\end{tabularx}
\end{table}

\subsection{Clustering} \label{sec:exp-quant-clusters}
We quantitatively evaluate the proposed approach by comparing it against a wide range of state-of-the-art community detection methods on multiple heterogeneous graphs from a variety of application domains with fundamentally different structural properties~\cite{nr}. 
{\smallskip
\begin{compactenum}[$\bullet$ \leftmargin=0em]
\item \textbf{Densest Subgraph} (\textbf{DS-H})~\cite{khuller2009finding}:
This baseline finds an approximation of the densest subgraph in $G$ using degeneracy ordering~\cite{erdHos1966chromatic,rossi14coloring-networks}.
Given a graph $G$ with $n$ nodes, let $H_i$ be the subgraph induced by $i$ nodes. 
At the start, $i=n$ and thus $H_i =  G$.
At each step, node $v_i$ with smallest degree is selected from $H_i$ and removed to obtain $H_{i-1}$.
Afterwards, we update the corresponding degrees of $H_{i-1}$ and density $\rho(H_{i-1})$.
This is repeated to obtain $H_{n},H_{n-1},\ldots,H_{1}$.
From $H_{n},H_{n-1},\ldots,H_{1}$, we select the subgraph $H_k$ with maximum density $\rho(H_k)$.

\item \textbf{KCore Communities} (\textbf{KCore-H})~\cite{rossi2015pmc-sisc,shin2016corescope}:
Many have observed the maximum k-core subgraph of a real-world network to be a highly dense subgraph that often contains the maximum clique~\cite{rossi2015pmc-sisc}. 
The KCore baseline simply uses the maximum k-core subgraph as $S$ and $\bar{S} = V \setminus S$.

\item \textbf{Label Propagation} (\textbf{LP-H})~\cite{raghavan2007near}:
Label propagation takes a labeling of the graph (in this case, induced by the heterogeneous graph), then for each node in some random ordering of the nodes, the node label is updated according to the label with maximal frequency among its neighbors. This iterative process converges when each node has the same label as the maximum neighboring label, or the number of iterations can be fixed. 
The final labeling induces a partitioning of the graph and the partition with maximum modularity is selected.

\item \textbf{Louvain} (\textbf{Louv-H})~\cite{blondel2008fast}:
Louvain performs a greedy optimization of modularity by forming small, locally optimal communities then grouping each community into one node. It iterates over this two-phase process until modularity cannot be maximized locally. 
The community with maximum modularity is selected.

\item \textbf{Spectral Clustering}  (\textbf{Spec-H})~\cite{chung1997spectralbook}:
This baseline executes spectral clustering on the normalized Laplacian of the adjacency matrix to greedily build the sweeping cluster that minimizes conductance. 

\item \textbf{Untyped-Graphlet Spec. Clustering} (\textbf{GSpec-H})~\cite{benson2016higher}:
This baseline computes the untyped-graphlet adjacency matrix and executes spectral clustering on the normalized Laplacian of this matrix to greedily build the sweeping cluster that minimizes the untyped-graphlet conductance.

\end{compactenum}\smallskip}\noindent
Note that we append the original method name with $-\mathbf{H}$ to indicate that it was adapted to support community detection in arbitrary heterogeneous graphs (Figure~\ref{fig:coupled-matrix-tensor}) since the original methods were not designed for such graph data.

\begin{table}[b!]
\vspace{-2mm}
\centering
\setlength{\tabcolsep}{3.5pt}
\renewcommand{\arraystretch}{1.2} 
\scriptsize
\caption{Gain/loss achieved by TGS over the other methods. 
Overall, TGS achieves a mean improvement of 43.53x over all graph data and baseline methods.
Note the last column reports the mean improvement achieved by TGS over all methods for each graph whereas the last row reports the mean improvement achieved by TGS over all graphs for each method.}
\vspace{-3mm}
\label{table:quant-cond-gain}
\small
\fontsize{8.0}{9.0}\selectfont
\begin{tabularx}{1.0\linewidth}{@{} 
l
XX XXXXr
HH
HH
}
\toprule
& & & & & & & \textsc{Mean} \\
&  \textbf{DS} &  \textbf{KC} &  \textbf{LP} &  \textbf{Louv} &  \textbf{Spec}     &  \textbf{GSpec}   & \textsc{Gain} \\ 
\midrule
\textsf{yahoo-msg}  &  9.69x  &  11.27x  &  3.98x  &  5.59x  &  1.22x  &  3.40x  &  	\textbf{5.86x}  &  \\
\textsf{dbpedia}  &  29.78x  &  22.43x  &  18.08x  &  33.14x  &  39.01x  &  37.77x  &  	\textbf{30.03x}  &  \\
\textsf{digg}  &  1030x  &  1110x  &  1888x  &  808x  &  1.50x  &  1.00x  &  	\textbf{806.75x}  &  \\
\textsf{movielens}  &  1.81x  &  1.93x  &  1.54x  &  1.72x  &  2.00x  &  1.20x  &  	\textbf{1.70x}  &  \\
\textsf{citeulike}  &  1.38x  &  1.39x  &  1.34x  &  1.21x  &  1.39x  &  1.39x  &  	\textbf{1.35x}  &  \\
\textsf{fb-CMU}  &  1.35x  &  1.91x  &  1.98x  &  1.75x  &  1.14x  &  1.71x  &  	\textbf{1.64x}  &  \\
\textsf{reality}  &  95.24x  &  39.19x  &  29.02x  &  19.92x  &  75.34x  &  2.05x  &  	\textbf{43.46x}  &  \\
\textsf{gene}  &  5.67x  &  6.51x  &  6.40x  &  5.84x  &  2.94x  &  1.17x  &  	\textbf{4.75x}  &  \\
\textsf{citeseer}  &  15.02x  &  20.02x  &  20.42x  &  18.68x  &  1.58x  &  1.58x  &  	\textbf{12.88x}  &  \\
\textsf{cora}  &  10.00x  &  13.33x  &  17.22x  &  16.36x  &  1.74x  &  1.74x  &  	\textbf{10.07x}  &  \\
\textsf{webkb}  &  1.00x  &  4.18x  &  2.77x  &  3.89x  &  3.00x  &  1.50x  &  	\textbf{2.72x}  &  \\
\textsf{pol-retweet}  &  25.85x  &  29.51x  &  1.32x  &  4.17x  &  30.35x  &  1.91x  &  	\textbf{15.52x}  &  \\
\textsf{webkb-spam}  &  1.71x  &  1.86x  &  1.96x  &  1.47x  &  1.98x  &  1.06x  &  	\textbf{1.67x}  &  \\
\textsf{fb-relationship}  &  1.26x  &  2.00x  &  1.08x  &  1.69x  &  2.01x  &  1.18x  &  	\textbf{1.54x}  &  \\
\textsf{Enzymes-g123}  &  2.22x  &  2.28x  &  1.46x  &  1.63x  &  1.73x  &  1.33x  &  	\textbf{1.78x}  &  \\
\textsf{Enzymes-g279}  &  3.64x  &  1.62x  &  1.89x  &  1.62x  &  1.07x  &  1.20x  &  	\textbf{1.84x}  &  \\
\textsf{Enzymes-g293}  &  3.31x  &  1.61x  &  3.12x  &  1.25x  &  2.54x  &  1.65x  &  	\textbf{2.25x}  &  \\
\textsf{Enzymes-g296}  &  3.11x  &  2.20x  &  2.89x  &  2.29x  &  1.99x  &  1.56x  &  	\textbf{2.34x}  &  \\
\textsf{NCI109-g4008}  &  2.38x  &  1.35x  &  2.64x  &  2.57x  &  1.09x  &  1.42x  &  	\textbf{1.91x}  &  \\
\textsf{NCI109-g1709}  &  3.00x  &  2.38x  &  1.07x  &  3.46x  &  2.94x  &  2.74x  &  	\textbf{2.60x}  &  \\
\textsf{NCI109-g3713}  &  2.04x  &  1.90x  &  3.75x  &  2.29x  &  3.33x  &  5.00x  &  	\textbf{3.05x}  &  \\
\textsf{NCI1-g3700}  &  1.63x  &  1.52x  &  3.00x  &  1.83x  &  1.33x  &  2.67x  &  	\textbf{2.00x}  &  \\
\midrule
\textsc{Mean Gain}  &  56.89x  &  58.23x  &  91.62x  &  42.74x  &  8.24x  &  3.47x  &  	\textbf{(43.53x)}  &  \\
\bottomrule
\end{tabularx}
\end{table}

\begin{figure}[h!]
\centering
\subfigure{\includegraphics[width=0.94\linewidth]{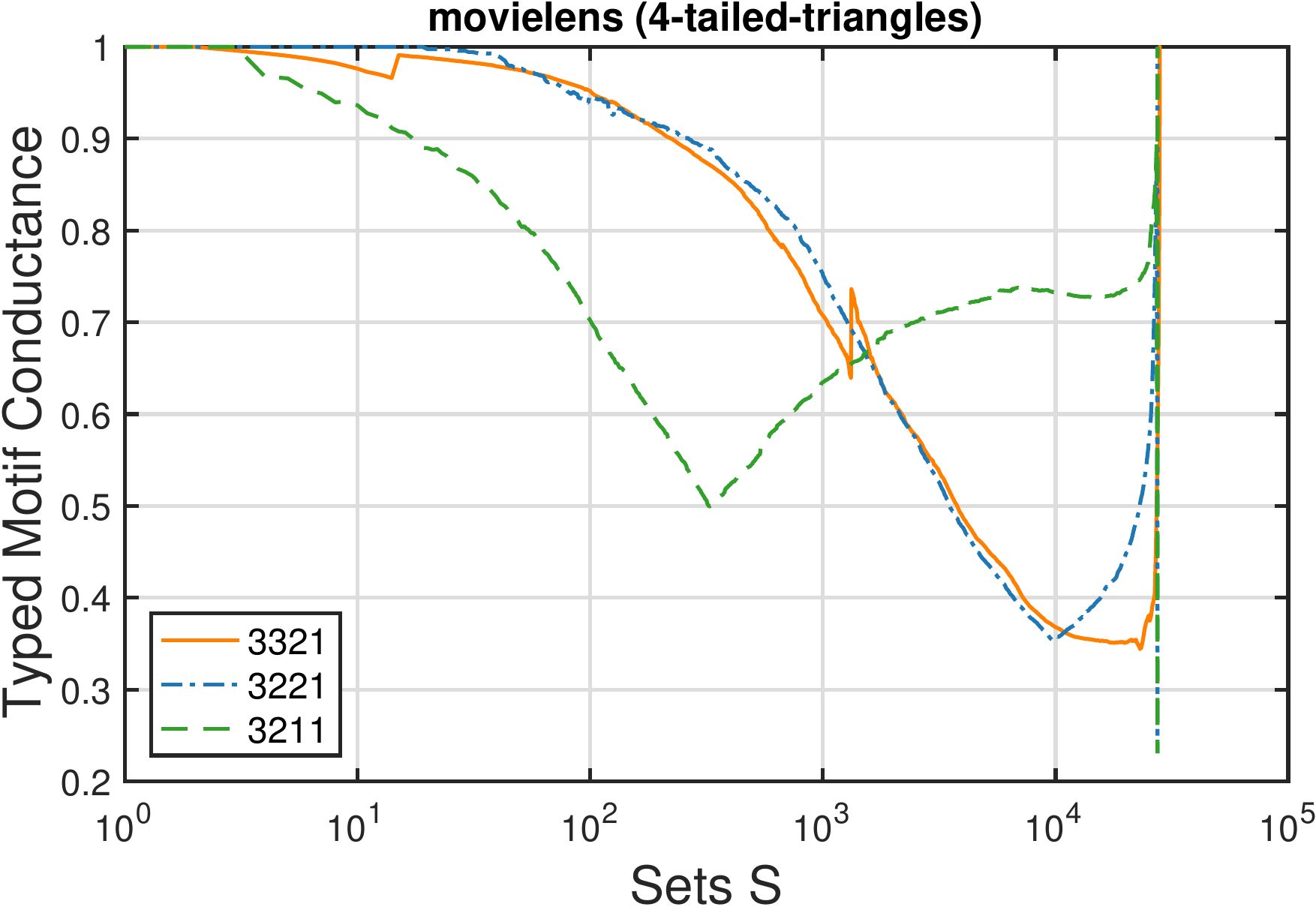}}
\subfigure{\includegraphics[width=0.94\linewidth]{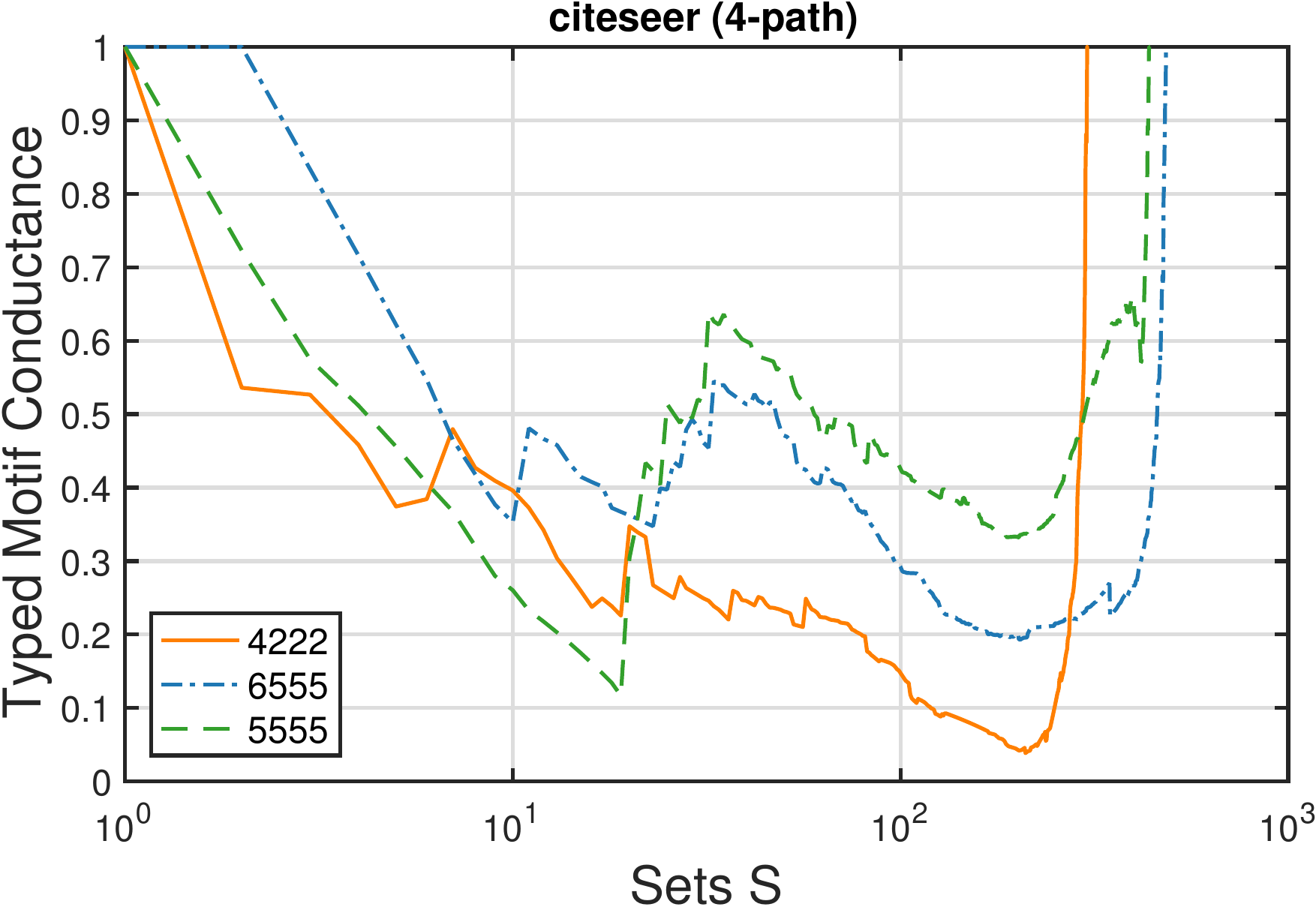}}
\caption{Typed graphlet conductance as a function of S from the sweep in Algorithm~\protect\ref{algo:spectral} for a variety of typed 4-path graphlets and typed 4-tailed-triangle graphlets. For this experiment, we consider the largest connected component for the graph derived from each typed graphlet.}
\label{fig:typed-graphlet-cond-vs-cluster-size}
\end{figure}

We evaluate the quality of communities using their external conductance score~\cite{gleich2012vertex,almeida2011there}.
This measure has been identified as one of the most important cut-based measures in a seminal survey by Schaeffer~\cite{Schaeffer2007} and extensively studied in many disciplines and applications~\cite{shi2000normalized, voevodski2009finding, chung1997spectralbook, gleich2012vertex, almeida2011there, Kannan2004, Schaeffer2007}.
Results are reported in Table~\ref{table:quant-conductance}.
As an aside, all methods take as input the same heterogeneous graph $G$.
Overall, the results in Table~\ref{table:quant-conductance} indicate that the proposed approach is able to reveal better high quality clusters across a wide range of heterogeneous graphs.
The heterogeneous network statistics and properties including the number of unique typed motifs for each induced subgraph pattern is shown in Table~\ref{table:network-stats}.

We also provide the improvement (gain) achieved by TGS clustering over the other methods in Table~\ref{table:quant-cond-gain}.
Note improvement is simply $\frac{\mathbb{E}(\mathcal{A}_i)}{\mathbb{E}(\mathcal{A}_{*})}$ where $\mathbb{E}(\mathcal{A}_i)$ is the external conductance of the solution given by algorithm $\mathcal{A}_i$ and $\mathcal{A}_{*}$ denotes the TGS algorithm.
Values less than 1 indicate that TGS performed worse than the other method whereas values $>1$ indicate the improvement factor achieved by TGS.
Overall, TGS achieves a mean improvement of $43.53$x over all graph data and baseline methods (Table~\ref{table:quant-cond-gain}).
Note the last column of Table~\ref{table:quant-cond-gain} reports the mean improvement achieved by TGS over all methods for each graph whereas the last row reports the mean improvement achieved by TGS over all graphs for each method.
Figure~\ref{fig:typed-graphlet-cond-vs-cluster-size} shows how \emph{typed graphlet conductance} (Eq.~\ref{eq:typed-graphlet-conductance}) changes as a function of community size $|S|$ for three different typed-graphlets.

\subsection{Link Prediction in Heterogeneous Graphs} \label{sec:exp-link-pred}
This section quantitatively demonstrates the effectiveness of TGS for link prediction.

\subsubsection{Higher-order Typed-Graphlet Embeddings}
In Section \ref{sec:exp-quant-clusters} we used the approach for higher-order clustering and quantitatively evaluated the quality of them.
In this section, we use the approach proposed in Section~\ref{sec:framework} to derive higher-order typed-graphlet node embeddings and quantitatively evaluate them for link prediction.
Algorithm~\ref{alg:higher-order-typed-graphlet-embeddings} summarizes the method for deriving higher-order typed motif-based node embeddings (as opposed to clusters/partitions of nodes, or an ordering for compression/analysis, see Section~\ref{sec:exp-graph-compression}).
In particular, given a typed-graphlet adjacency matrix, Algorithm~\ref{alg:higher-order-typed-graphlet-embeddings} outputs a $N \times D$ matrix $\mZ$ of node embeddings.
For graphs with many connected components, Algorithm~\ref{alg:higher-order-typed-graphlet-embeddings} is called for each connected component of $G^{H}$ and the resulting embeddings are stored in the appropriate locations in the overall embedding matrix $\mZ$.

\renewcommand{\GraphletFigScale}{0.05}
\begin{table}[h!]
\centering
\renewcommand{\arraystretch}{1.1} 
\caption{
Link prediction edge types and semantics.
We bold the edge type that is predicted by the models.
}
\label{table:link-pred-network-data}
\vspace{-3mm}
\fontsize{7.5}{8.5}\selectfont
\begin{tabularx}{1.0\linewidth}{l       HHHHcH    HHHHHHHH  HHHH Hl HHHHH@{}}
\toprule
\textbf{Graph}  &   
$|V|$  &  $|E|$  && & $|\mathcal{T}_V|$ &  &
\includegraphics[scale=0.8]{fig2.pdf} &
\includegraphics[scale=0.8]{fig3.pdf} &
\includegraphics[scale=0.15]{fig4.pdf} &
\includegraphics[scale=0.8]{fig5.pdf} &
\includegraphics[scale=0.8]{fig6.pdf} &
\includegraphics[scale=0.15]{fig7.pdf} &
\includegraphics[scale=0.14]{fig8.pdf} &
\includegraphics[scale=0.8]{fig9.pdf} &
&&&&
\textbf{Predicted Edge Type} &
\textbf{Heterogeneous Edge Types} &
\\
\midrule

\textsf{movielens} & 28.1k & 170.4k & && 3 & 3 & 
7 & 1 & 6 & 9 & 6 & 3 & 3 & 0 & 
&&&& 
& 
\textbf{user-by-movie}, user-by-tag & 
\\ 

&&&&&&&
&&&&&&&&
&&&&
&
tag-by-movie &
\\

\textsf{dbpedia} &  495.9k & 921.7k & && 4 & 3 &
8 & 0 & 6 & 10 & 5 & 0 & 0 & 0 & 
&&&& 
&
\textbf{person-by-work} (produced work),  
\\

&&&&&&&
&&&&&&&&
&&&&
&
person-has-occupation, & 
\\

&&&&&&&
&&&&&&&&
&&&&
&
work-by-genre (work-associated-with-genre) & 
\\

\textsf{yahoo-msg} & 100.1k & 739.8k & && 2 & 2 & 
3 & 2 & 3 & 4 & 3 & 3 & 3 & 2 & 
&&&& 
user-user communication & 
\textbf{user-by-user} (communicated with), &
\\ 

&&&&&&&
&&&&&&&&
&&&&
&
user-by-location (communication location) &
\\

\bottomrule
\end{tabularx}
\end{table}

\subsubsection{Experimental Setup}
We evaluate the higher-order typed-graphlet node embedding approach (Algorithm~\ref{alg:higher-order-typed-graphlet-embeddings}) against the following methods: DeepWalk (DW)~\cite{deepwalk}, LINE~\cite{line}, GraRep~\cite{grarep}, spectral embedding (untyped edge motif)~\cite{Ng2002}, and spectral embedding using untyped-graphlets.
All methods output ($D$=128)-dimensional node embeddings $\mZ = \big[\, \vz_1 \cdots \vz_n \,\big]^T$ where $\vz_i \in \RR^{D}$.
For DeepWalk (DW)~\cite{deepwalk}, we perform 10 random walks per node of length 80 as mentioned in~\cite{node2vec}.
For LINE~\cite{line}, we use 2nd-order proximity and perform 60 million samples.
For GraRep (GR)~\cite{grarep}, we use $K=2$.
In contrast, the spectral embedding methods do not have any hyperparameters besides $D$ which is fixed for all methods.
As an aside, all methods used for comparison were modified to support heterogeneous graphs (similar to how the other baseline methods from Section~\ref{sec:exp-quant-clusters} were modified).
In particular, we adapted the methods to allow multiple graphs as input consisting of homogeneous or bipartite graphs that all share at least one node type (See Table~\ref{table:link-pred-network-data} and Figure~\ref{fig:coupled-matrix-tensor}) and from these graphs we construct a single large graph by simply ignoring the node and edge types and relabeling the nodes to avoid conflicts.

\subsubsection{Comparison}
Given a partially observed graph $G$ with a fraction of missing/unobserved edges, the link prediction task is to predict these missing edges.
We generate a labeled dataset of edges.
Positive examples are obtained by removing $50\%$ of edges uniformly at random, whereas \emph{negative examples} are generated by randomly sampling an equal number of node pairs $(i,j) \not\in E$. 
For each method, we learn embeddings using the remaining graph.
Using the embeddings from each method, we then learn a logistic regression (LR) model to predict whether a given edge in the test set exists in $E$ or not.
Experiments are repeated for 10 random seed initializations and the average performance is reported.
All methods are evaluated against four different evaluation metrics including $F_1$, Precision, Recall, and AUC.

\begin{table}[h!]
\centering
\renewcommand{\arraystretch}{1.2}
\caption{Link prediction results.}
\label{table:link-pred-results}
\vspace{-3mm}
\fontsize{7.5}{8.5}\selectfont
\begin{tabularx}{1.0\linewidth}{@{}ll XXXXX X@{}H HHHH@{}}
\toprule
&&
\textbf{DW}   &   \textbf{LINE}   &   \textbf{GR}   &  \textbf{Spec} &  \textbf{GSpec} &  \textbf{TGS} && \\ 
\midrule

\multirow{4}{*}{\rotatebox{70}{\textbf{\sf movielens}}}  
&   $\mathbf{F_1}$          &   	0.8544   &   0.8638       &   0.8550        &   0.8774   &   0.8728   &   \textbf{0.9409}   &   \\
&   \textbf{Prec.}          &   	0.9136   &   0.8785       &   0.9235        &   0.9409   &   0.9454   &   \textbf{0.9747}   &   \\
&   \textbf{Recall}         &   	0.7844   &   0.8444       &   0.7760        &   0.8066   &   0.7930   &   \textbf{0.9055}   &   \\
&   \textbf{AUC}            &   	0.9406   &   0.9313       &   0.9310        &   0.9515   &   0.9564   &   \textbf{0.9900}   &   \\
\midrule

\multirow{4}{*}{\rotatebox{70}{\textbf{\sf dbpedia}}} 
&   $\mathbf{F_1}$          &   	0.8414   &   0.7242   &   0.7136       &   0.8366   &   0.8768   &   \textbf{0.9640}   &   \\
&   \textbf{Prec.}          &   	0.8215   &   0.7754   &   0.7060       &   0.7703   &   0.8209   &   \textbf{0.9555}   &   \\
&   \textbf{Recall}         &   	0.8726   &   0.6375   &   0.7323       &   0.9669   &   0.9665   &   \textbf{0.9733}   &   \\
&   \textbf{AUC}            &   	0.8852   &   0.8122   &   0.7375       &   0.9222   &   0.9414   &   \textbf{0.9894}   &   \\
\midrule

\multirow{4}{*}{\rotatebox{70}{\textbf{\bf \sf yahoo}}}
&   $\mathbf{F_1}$          &   	0.6927   &   0.6269   &   0.6949   &   0.9140    &   0.8410  &   \textbf{0.9303}   &     &   \\
&   \textbf{Prec.}          &   	0.7391   &   0.6360  &   0.7263   &   0.9346   &   0.8226   &   \textbf{0.9432}   &    &   \\
&   \textbf{Recall}         &   	0.5956   &   0.5933   &   0.6300   &   0.8904   &   0.8699   &   \textbf{0.9158}   &   &   \\
&   \textbf{AUC}            &   	0.7715   &   0.6745   &   0.7551   &   0.9709   &   0.9272   &   \textbf{0.9827}   &    &   \\

\bottomrule
\multicolumn{9}{l}{\footnotesize $^{\star}$ Note DW=DeepWalk and GR=GraRep.}\\
\end{tabularx}
\end{table}

For link prediction~\cite{linkpred:liben07,role2vec}, entity resolution/network alignment, recommendation and other machine learning tasks that require edge embeddings (features)~\cite{deepGL}, we derive edge embedding vectors by combining the learned node embedding vectors of the corresponding nodes using an edge embedding function $\Phi$.
More formally, given $D$-dimensional embedding vectors $\vz_i$ and $\vz_j$ for node $i$ and $j$, 
we derive a $D$-dimensional edge embedding vector 
$\vz_{ij} = \Phi(\vz_i, \vz_j)$
where $\Phi$ is defined as one of the following \emph{edge embedding functions}:
\begin{align}\label{eq:edge-embedding-ops}\nonumber
\Phi \,\in\, 
\Bigg\lbrace
\frac{\vz_i + \vz_j}{2},\;
\vz_i \odot \vz_j,\;
\abs{\vz_i - \vz_j},\;
(\vz_i - \vz_j)^{\circ 2},\,
\max(\vz_i, \vz_j),\,
\vz_i + \vz_j
\Bigg\rbrace
\end{align}\noindent
Note $\vz_i \odot \vz_j$ is the element-wise (Hadamard) product, $\vz^{\circ 2}$ is the Hadamard power, and $\max(\vz_i, \vz_j)$ is the 
element-wise max.

\begin{figure*}[h!]
\centering
\subfigure[Spectral clustering (untyped edge)]{\includegraphics[width=0.30\linewidth]{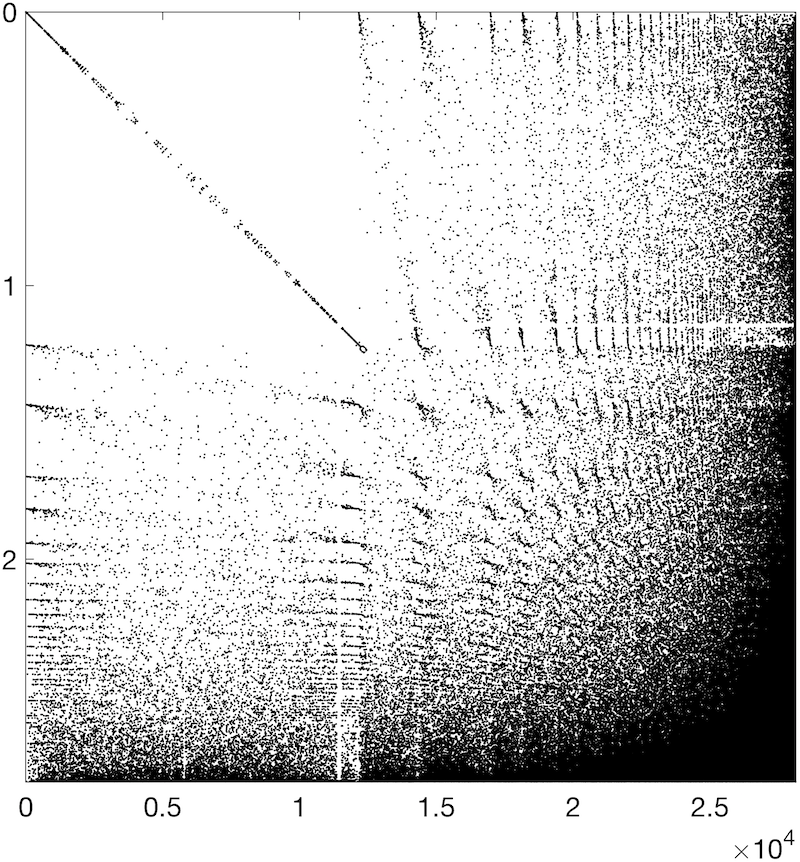}\label{fig:spectral-spy-plot-movielens}}
\hfill
\subfigure[Untyped 3-path]{\includegraphics[width=0.30\linewidth]{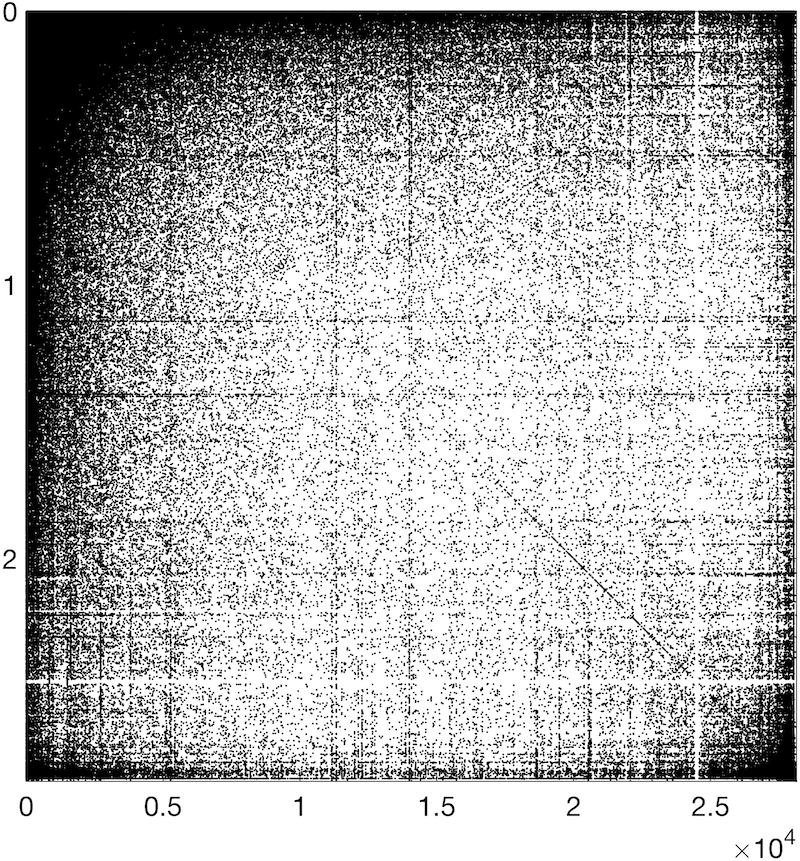}\label{fig:untyped-3path-spy-plot-movielens}}
\hfill
\subfigure[Typed 4-cycle ]{\includegraphics[width=0.334\linewidth]{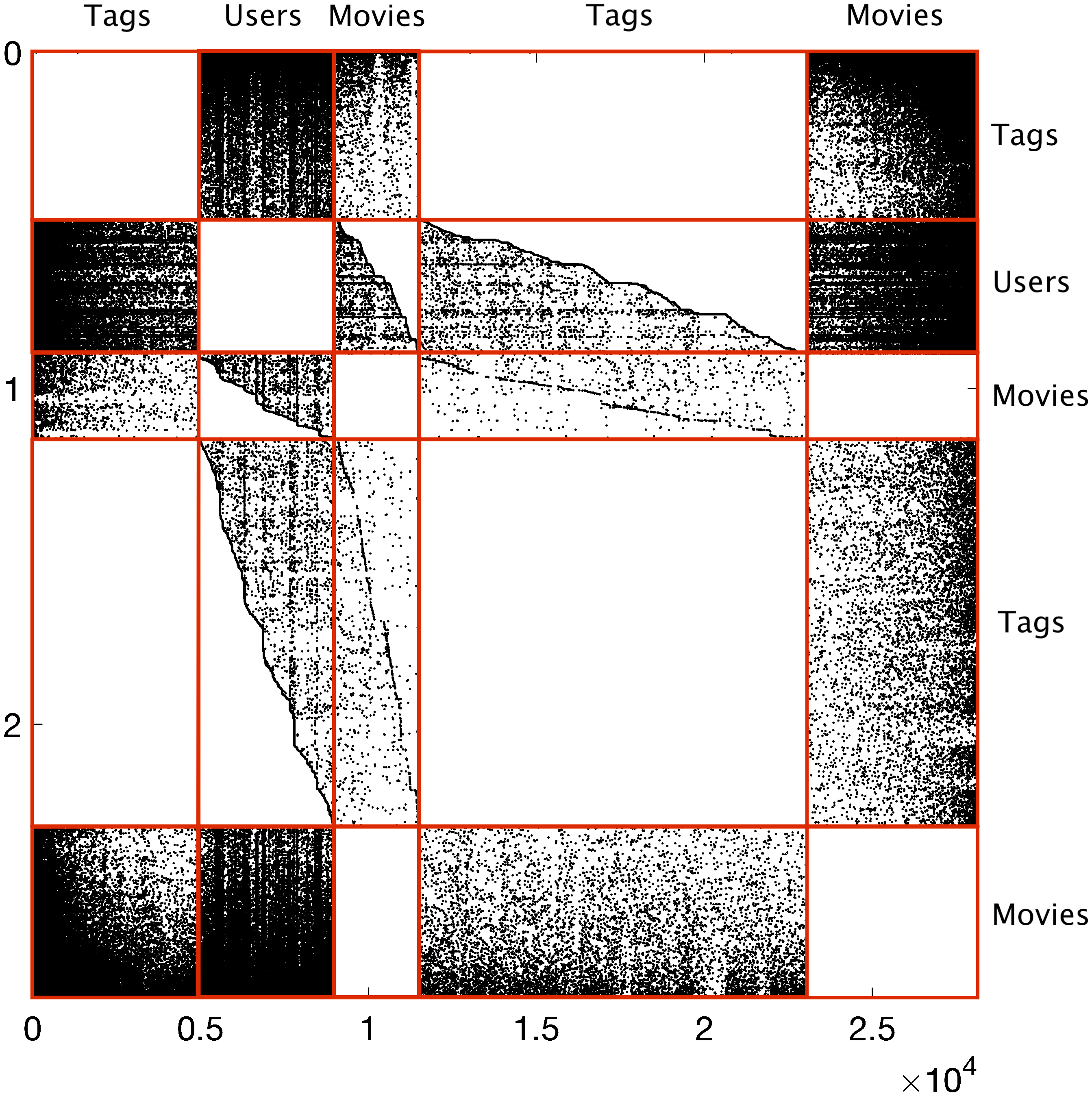}
\label{fig:typed-4cycle-spy-plot-movielens}}

\vspace{-2mm}
\caption{Typed graphlet-based spectral ordering achieves significant compression 
by partitioning users, tags, and movies (from the movielens data) into homogeneous groups that are either nearly fully-connected (near-clique) or fully-disconnected.
Strikingly, TGS partitions the rows/columns according to types without explicitly leveraging types (\ie, types are not used when deriving the typed-graphlet spectral ordering).
For instance, the first $\approx$5k rows/columns correspond to tags, whereas the following $\approx$4k rows/columns are users, and so on.
This is in contrast to the other methods where the rows/columns of different types are mixed with one another in a seemingly random fashion.
Moreover, these approaches fail to partition the graph into homogeneous groups that are dense or completely empty.
The typed 4-cycle graphlet used above consists of 2 nodes representing movies and the other two representing tags assigned to the movies.
Other typed-graphlets gave other interesting results with comparable compression.}
\label{fig:spy-motif-adj-vs-others-movielens}
\end{figure*}

Table~\ref{table:link-pred-network-data} summarizes the heterogeneous network data used for link prediction.
In particular, the types used in each of the heterogeneous networks are shown in Table~\ref{table:link-pred-network-data} as well as the specific types involved in the edges that are predicted (\eg, the edge type being predicted).
The results are provided in Table~\ref{table:link-pred-results}.
Results are shown for the best edge embedding function.
In Table~\ref{table:link-pred-results}, TGS is shown to outperform all other methods across \emph{all} four evaluation metrics.
In all cases, the higher-order typed-graphlet spectral embedding outperforms the other methods (Table~\ref{table:link-pred-results}) with an overall mean gain (improvement) in $F_1$ of 18.7\% (and up to 48.4\% improvement) across all graph data.
In terms of AUC, TGS achieves a mean gain of 14.4\% (and up to 45.7\% improvement) over all methods.
We posit that an approach similar to the one proposed in~\cite{HONE} could be used with the typed-graphlet node embeddings to achieve even better predictive performance.
This approach would allow us to leverage multiple typed-graphlet Laplacian matrices for learning more appropriate higher-order node embeddings.

\subsection{Graph Compression} \label{sec:exp-graph-compression}
In Section \ref{sec:exp-quant-clusters} we used the approach for higher-order clustering whereas Section~\ref{sec:exp-link-pred} demonstrated the effectiveness of the approach for link prediction.
However, the framework can be leveraged for many other important applications including graph compression~\cite{bvgraph,boldi2011layered,chierichetti2009compressing,rossi2018compressing-graphs-cliques,liakos2014pushing}.
In this section, we explore the proposed approach for graph compression.
Compression has two key benefits.
First, it reduces the amount of IO traffic~\cite{rossi2018compressing-graphs-cliques}.
Second, it can speed up existing algorithms by reducing the amount of work required~\cite{karande2009speeding}.
Graph compression methods rely on a ``good'' ordering of the vertices in the graph to achieve a good compression~\cite{bvgraph,boldi2011layered}.

\begin{table}[h!]
\centering
\renewcommand{\arraystretch}{1.2}
\caption{Graph compression results. Size in bytes required to store heterogeneous graphs using the bvgraph compression scheme with different orderings.}
\label{table:graph-compression}
\vspace{-3mm}
\fontsize{7.5}{8.5}\selectfont
\begin{tabularx}{1.0\linewidth}{@{}l XXXX HH X@{}H HHHH@{}}
\toprule
& \multicolumn{4}{c}{\sc Bytes} && \textsc{Mean}\\
\cmidrule(l{5pt}r{10pt}){2-5}

\textsc{Graph} &  \textbf{Native} &  \textbf{Spec} &  \textbf{GSpec} &  \textbf{TGS} &&& \textsc{Gain} \\ 
\midrule
\textsf{movielens}  &  585588  &  471246  &  464904  &  444252  &&& \textbf{14.18\%}  &  \\
\textsf{yahoo-msg}  &  3065499  &  2694151  &  2708700  &  2427325  &&& \textbf{16.29\%}  &  \\
\textsf{dbpedia}  &  4800584  &  3520721  &  3469878  &  3111728  &&& \textbf{26.31\%}  &  \\
\textsf{digg}  &  15989475  &  10462874  &  10296144  &  9677741  &&& \textbf{26.57\%}  &  \\
\bottomrule
\end{tabularx}
\end{table}

In this work, we order the vertices by the \emph{typed-graphlet spectral ordering} introduced previously in Definition~\ref{def:typed-graphlet-spectral-ordering}.
Notice in this case, the output of our approach is the typed-graphlet spectral ordering (Definition~\ref{def:typed-graphlet-spectral-ordering}) as opposed to clusters (Section~\ref{sec:exp-quant-clusters}) or node embeddings (Section~\ref{sec:exp-link-pred}).
We then evaluate how well the bvgraph~\cite{bvgraph} compression method reduces the graph size using this ordering.
Given an ordering, we permute the graph to use this ordering and use the bvgraph compression algorithm~\cite{bvgraph} with all the default settings to compress the networks.
Results are reported in Table~\ref{table:graph-compression} for four large heterogeneous graphs.
We compare the compression obtained by reporting the size of each heterogeneous graph in bytes after compression.
We evaluate four orderings of the vertices: the native order, spectral ordering (untyped edge), untyped-graphlet ordering and the typed-graphlet spectral ordering proposed in this work.
For untyped-graphlet and typed-graphlet spectral ordering we report the best result given by an ordering from any untyped or typed-graphlet.
We find that the typed-graphlet ordering results in better compression across all other methods and graphs. 
Overall, typed-graphlet spectral ordering achieves a mean improvement of $20.8\%$ over all graphs and all orderings.

In addition to the quantitative compression results shown in Table~\ref{table:graph-compression}, we use the proposed ordering for exploratory analysis.
In particular, we use the orderings to permute the rows/columns of the original adjacency matrix and visualize the nonzero structure of the resulting matrices in Figure~\ref{fig:spy-motif-adj-vs-others-movielens}.
Using the ordering from TGS (Definition~\ref{def:typed-graphlet-spectral-ordering}) gives rise to partitions (sub-matrices) that are significantly more homogeneous (completely connected or disconnected) than the other methods as shown in Figure~\ref{fig:spy-motif-adj-vs-others-movielens} and thus are able to achieve a better compression as shown quantitatively in Table~\ref{table:graph-compression}.
Furthermore, TGS is able to uncover the type of the nodes by grouping nodes into partitions based on their types (users, movies, tags).
Moreover, the partitions are also meaningful as they partition movies and the tags used to describe those movies into genres. 
This allows us to understand the tags that best describe that genre as well as the movies from that genre that align with those tags.
Other typed-graphlet spectral orderings from different typed graphlets were removed due to space, though many of them also gave interesting and explainable block partitions as well.

\renewcommand{\subsection}[1]{\smallskip\noindent\textbf{#1}.}
\section{Related Work} \label{sec:related-work}

\subsection{\textit{Community Detection in Homogeneous Graphs}} \label{sec:related-work-homo}
Most research in community detection has traditionally focused on homogeneous graphs \cite{Bothorel2015}. This problem has been extensively researched as evidenced by the multiple survey papers \cite{Schaeffer2007, Fortunato2010, Newman2011, Coscia2011, Malliaros2013, Fortunato2016} and empirical comparisons of algorithms \cite{Arenas2005, Lancichinetti2009, Leskovec2010, Harenberg2014} on this topic.
Many works have focused on community detection techniques using modularity-based optimization.
Modularity was introduced in the seminal paper \cite{Newman2004} as a quantitative measure for assigning scores to a community structure. It became a standard measure for comparing clustering algorithms and suggested a framework for community detection as an optimization task of a quality function. Modularity maximization is an NP-hard problem, but there exist efficient heuristics such as greedy methods, semidefinite programming, simulated annealing, and spectral methods \cite{Newman2011}. Nevertheless, it suffers from many drawbacks such as runtime dependence on the size of graph \cite{Fortunato2016}, resolution limit \cite{Fortunato2007}, and nonuse of between-community connectivity information \cite{Newman2011}.

Graph conductance is another very popular community quality function~\cite{chung1997spectralbook, Kannan2004}.
Computing the conductance of a graph is an NP-hard problem as well \cite{vsima2006np}, but there exist spectral methods that give good, \emph{theoretically-supported} approximations~\cite{chung1997spectralbook, Leighton1999, Ng2002, Verma2003, Kannan2004} and have only weak runtime dependence on the size of the graph since there exist fast methods for computing eigenvalues \cite{Golub2012}. 
Moreover, conductance takes into account the internal and external connectedness of a community \cite{Fortunato2016}.
As an aside, existing higher-order clustering methods that extend modularity~\cite{Gomez2008} and conductance~\cite{benson2016higher,Tsourakakis2017} are all designed for \emph{homogeneous graphs} with a single node/edge type \emph{and} are also based on the existing notion of \emph{untyped graphlets}. 
However, we discuss these methods along side other higher-order methods that leverage untyped graphlets.
In this work, we propose a higher-order clustering framework that generalizes to heterogeneous graphs.
However, since homogeneous graphs are a special case of heterogeneous graphs (where nodes/edges have a single type), the proposed framework can be used for higher-order clustering in homogeneous graphs as well.

\subsection{\textit{Community Detection in Heterogeneous Graphs}} \label{sec:related-work-heter}
Recently, researchers have started to extend community detection methods for multi-relational, multi-typed graphs \cite{Bothorel2015, Shi2017}. In the literature, these graphs are referred to as \emph{heterogeneous graphs} or \emph{heterogeneous information networks} \cite{Shi2017}. In recent years, many methods have been proposed for community detection in heterogeneous networks in ways that consolidate both structural and compositional information.
Weight modification methods reduce a heterogeneous graph to a weighted homogeneous graph through an edge-weighting function based on node types. Afterwards, any homogeneous community detection algorithm can be applied to this modified graph. For example, \cite{Neville2003} and \cite{Steinhaeuser2008} use a matching coefficient function that quantifies the number of similar node types. Under a certain viewpoint, our method can be classified under this category of algorithms. We discuss this later.

A different paradigm for combining both structural and compositional information of a network would be to take the opposite approach of weight modification. One could transform a heterogeneous graph to a point cloud by converting structural information coupled with node type data into a node distance function ~\cite{Shi2017}. Then, any distance-based clustering method such as k-means can be applied on this point cloud. This approach incorporates both structure and composition of the network. Linear combination methods take a linear combination of type similarity and structural similarity functions as proposed in \cite{Combe2012}. Walk strategies on heterogeneous graphs have also been used to compute vertex distance functions. The work in \cite{Zhou2009} defines a random walk on a heterogeneous graph such that more paths---alongside the paths from the network structure alone---exist between nodes of the same type, thus measuring vertex proximity with two modes of data. Another distance function based on breadth-first search is proposed in \cite{Ge2008} that uses the node types to determine the next visited node, and thus the distance.

These similarity reduction methods have a feature that nodes that are structurally far from each other but share similar attributes may become close after this modification \cite{Bothorel2015}. As a consequence, clusters may contain disconnected portions of the graph which is generally not seen as a characteristic of communities. Using motif-based clustering in our work allows us to preserve this connectedness property.
Another standard homogeneous clustering approach that has been extended to heterogeneous graphs is statistical inference such as generative models \cite{Liu2009}, stochastic block models \cite{Balasubramanyan2011}, and Bayesian inference \cite{Xu2012}. 

In this work, we develop a principled framework for \emph{higher-order clustering in heterogeneous graphs}.
Furthermore, while most existing methods for heterogeneous graphs lack a sound theoretically grounded framework, we rigorously prove mathematical guarantees on the optimality of the higher-order clustering obtained from the framework.

\subsection{\textit{Graphlets}} \label{sec:related-work-graphlets}
Graphlets (network motifs) were first introduced in \cite{Shen-Orr2002, Milo2002} to study the structural design principles of single-typed biological networks.
In that work, graphlets were found to be the fundamental building blocks of complex homogeneous (single-typed) networks.
Various algorithms have been developed to count the occurences of all graphlets up to a given size on the nodes~\cite{rage} and edges~\cite{pgd, pgd-kais} of a graph.
Motif discovery algorithms are limited in that they are computationally expensive for larger motifs and they search for motifs operating in isolation. 
In \cite{Knabe2008}, it was shown that network context, \ie, the connections of the motif to the rest of the network, is important in inferring the functionality of a motif. 
Motifs have recently been used in other higher order-network analysis methods such as role discovery \cite{Rossi2015}, network embeddings \cite{Rossi2018a}, inductive network representation learning~\cite{deepGL}, 
and temporal network analysis \cite{Paranjape2017}. 

Motifs were first used for community detection in \cite{Gomez2008}.
In that work, motif modularity was introduced as a generalization of the standard notion of modularity.
Once motif modularity is defined, their method essentially becomes a modularity maximization problem. 
As mentioned above, modularity-based methods ignore any between-community connectivity information. 
Therefore, the connectivity of the communities to the rest of the network is not considered in the motif-based modularity method, and thus we lose information that may be of value in correctly capturing useful community structure. 
Moreover, this method still suffers from the resolution limit \cite{Fortunato2007} and requires longer computation.
More recently, \cite{benson2016higher,Tsourakakis2017} extended the definition of graph conductance based on the existing notion of untyped-motifs for homogeneous graphs.
This definition is a special case of the proposed framework when untyped graphlets are used and the graph is homogeneous.

Previous work has focused entirely on \emph{untyped motifs/graphlets}~\cite{pgd,pgd-kais,ahmed16bigdata,Rossi2018b}.
In this work, we introduce the generalized notion of \emph{typed graphlets} and use this more powerful representation as a basis for higher-order clustering.
Typed graphlets generalize the notion of graphlets to rich heterogeneous networks as they explicitly capture the higher-order typed connectivity patterns in such networks.
Using this more appropriate and general notion, we develop a principled general higher-order clustering framework by introducing  \emph{typed-graphlet conductance} that generalizes the traditional conductance to higher-order structures in heterogeneous graphs.
Recall that homogeneous, labeled, signed, and attributed graphs are all special cases of heterogeneous graphs.
The framework provides mathematical guarantees on the optimality of the higher-order clustering obtained.
The theoretical results extend to typed graphlets of arbitrary size and avoids restrictive special cases required in prior work. 
In addition, existing work on higher-order motif-based methods have focused entirely on simple homogeneous graphs whereas our work focuses on rich heterogeneous networks with an arbitrary number of node and edge types (Figure~\ref{fig:coupled-matrix-tensor}).
Furthermore, while previous work on higher-order clustering was designed for homogeneous graphs \emph{and} untyped-graphlets, they also focused only on community detection whereas this work also leverages the proposed framework for deriving higher-order embeddings and graph compression based on the typed-graphlet spectral ordering.

\section{Conclusion}\label{sec:conc}
This work proposed a general framework for higher-order spectral clustering in heterogeneous graphs.
The framework explicitly incorporates heterogeneous higher-order information by counting typed graphlets that leverage node and edge types.
It is shown that typed-graphlets generalize the notion of graphlets to rich heterogeneous networks and that these explicitly capture the higher-order typed connectivity patterns in such networks.
Using these as a basis, we proposed the notion of \emph{typed-graphlet conductance} that generalizes the notion of conductance to higher-order structures in heterogeneous graphs.
Typed-graphlet conductance minimization, for a given typed graphlet, provides a cut in the heterogeneous graph that preserves instances of the typed graphlet in a balanced manner. 

The framework provides mathematical guarantees on the optimality of the higher-order clustering obtained.
The theoretical results extend to typed graphlets of arbitrary size and avoids restrictive special cases required in prior work.
The framework unifies prior work and serves as a basis for analysis of higher-order spectral clustering methods.
It was shown that spectral clustering and untyped-graphlet spectral clustering are special cases in the proposed framework.
The experiments demonstrated the effectiveness and utility of the proposed framework for three important tasks including
(i) clustering,
(ii) predictive modeling, and 
(iii) graph compression.
For these tasks, the approach was shown to outperform other state-of-the-art methods with a significant improvement in all cases.
The approach achieves an overall improvement in $F_1$ and AUC of 18.7\% and 14.4\% for link prediction whereas for graph compression it achieves a mean improvement of $20.8\%$ across all graphs and methods.
Finally, typed-graphlet spectral clustering is shown to uncover better clusters than state-of-the-art methods with a mean improvement of 43x over all graphs and methods.

\balance
\bibliographystyle{ACM-Reference-Format}
\bibliography{paper}


\begin{thebibliography}{00}


\ifx \showCODEN    \undefined \def \showCODEN     #1{\unskip}     \fi
\ifx \showDOI      \undefined \def \showDOI       #1{#1}\fi
\ifx \showISBNx    \undefined \def \showISBNx     #1{\unskip}     \fi
\ifx \showISBNxiii \undefined \def \showISBNxiii  #1{\unskip}     \fi
\ifx \showISSN     \undefined \def \showISSN      #1{\unskip}     \fi
\ifx \showLCCN     \undefined \def \showLCCN      #1{\unskip}     \fi
\ifx \shownote     \undefined \def \shownote      #1{#1}          \fi
\ifx \showarticletitle \undefined \def \showarticletitle #1{#1}   \fi
\ifx \showURL      \undefined \def \showURL       {\relax}        \fi
\providecommand\bibfield[2]{#2}
\providecommand\bibinfo[2]{#2}
\providecommand\natexlab[1]{#1}
\providecommand\showeprint[2][]{arXiv:#2}

\bibitem[\protect\citeauthoryear{Acar, Kolda, and Dunlavy}{Acar
  et~al\mbox{.}}{2011}]%
        {acar2011all}
\bibfield{author}{\bibinfo{person}{Evrim Acar}, \bibinfo{person}{Tamara~G
  Kolda}, {and} \bibinfo{person}{Daniel~M Dunlavy}.}
  \bibinfo{year}{2011}\natexlab{}.
\newblock \showarticletitle{All-at-once optimization for coupled matrix and
  tensor factorizations}.
\newblock \bibinfo{journal}{{\em arXiv:1105.3422\/}} (\bibinfo{year}{2011}).
\newblock


\bibitem[\protect\citeauthoryear{Ahmed, Willke, and Rossi}{Ahmed
  et~al\mbox{.}}{2016a}]%
        {ahmed16bigmine}
\bibfield{author}{\bibinfo{person}{Nesreen Ahmed}, \bibinfo{person}{Ted
  Willke}, {and} \bibinfo{person}{Ryan~A. Rossi}.}
  \bibinfo{year}{2016}\natexlab{a}.
\newblock \showarticletitle{Exact and Estimation of Local Edge-centric Graphlet
  Counts}. In \bibinfo{booktitle}{{\em KDD BigMine}}. \bibinfo{pages}{16}.
\newblock


\bibitem[\protect\citeauthoryear{Ahmed, Neville, Rossi, and Duffield}{Ahmed
  et~al\mbox{.}}{2015}]%
        {pgd}
\bibfield{author}{\bibinfo{person}{Nesreen~K. Ahmed}, \bibinfo{person}{Jennifer
  Neville}, \bibinfo{person}{Ryan~A. Rossi}, {and} \bibinfo{person}{Nick
  Duffield}.} \bibinfo{year}{2015}\natexlab{}.
\newblock \showarticletitle{Efficient Graphlet Counting for Large Networks}. In
  \bibinfo{booktitle}{{\em ICDM}}. \bibinfo{pages}{1--10}.
\newblock


\bibitem[\protect\citeauthoryear{Ahmed, Neville, Rossi, Duffield, and
  Willke}{Ahmed et~al\mbox{.}}{2016}]%
        {pgd-kais}
\bibfield{author}{\bibinfo{person}{Nesreen~K. Ahmed}, \bibinfo{person}{Jennifer
  Neville}, \bibinfo{person}{Ryan~A. Rossi}, \bibinfo{person}{Nick Duffield},
  {and} \bibinfo{person}{Theodore~L. Willke}.} \bibinfo{year}{2016}\natexlab{}.
\newblock \showarticletitle{Graphlet Decomposition: Framework, Algorithms, and
  Applications}.
\newblock \bibinfo{journal}{{\em KAIS\/}} (\bibinfo{year}{2016}),
  \bibinfo{pages}{689--722}.
\newblock


\bibitem[\protect\citeauthoryear{Ahmed, Rossi, Zhou, Lee, Kong, Willke, and
  Eldardiry}{Ahmed et~al\mbox{.}}{2018}]%
        {role2vec}
\bibfield{author}{\bibinfo{person}{Nesreen~K. Ahmed}, \bibinfo{person}{Ryan~A.
  Rossi}, \bibinfo{person}{Rong Zhou}, \bibinfo{person}{John~Boaz Lee},
  \bibinfo{person}{Xiangnan Kong}, \bibinfo{person}{Theodore~L. Willke}, {and}
  \bibinfo{person}{Hoda Eldardiry}.} \bibinfo{year}{2018}\natexlab{}.
\newblock \showarticletitle{Learning Role-based Graph Embeddings}. In
  \bibinfo{booktitle}{{\em arXiv:1802.02896}}.
\newblock


\bibitem[\protect\citeauthoryear{Ahmed, Willke, and Rossi}{Ahmed
  et~al\mbox{.}}{2016b}]%
        {ahmed16bigdata}
\bibfield{author}{\bibinfo{person}{Nesreen~K. Ahmed},
  \bibinfo{person}{Theodore~L. Willke}, {and} \bibinfo{person}{Ryan~A. Rossi}.}
  \bibinfo{year}{2016}\natexlab{b}.
\newblock \showarticletitle{Estimation of Local Subgraph Counts}. In
  \bibinfo{booktitle}{{\em IEEE BigData}}. \bibinfo{pages}{586--595}.
\newblock


\bibitem[\protect\citeauthoryear{Albert and Barab\'asi}{Albert and
  Barab\'asi}{2002}]%
        {Albert2002}
\bibfield{author}{\bibinfo{person}{R\'eka Albert} {and}
  \bibinfo{person}{Albert-L\'aszl\'o Barab\'asi}.}
  \bibinfo{year}{2002}\natexlab{}.
\newblock \showarticletitle{Statistical mechanics of complex networks}.
\newblock \bibinfo{journal}{{\em Rev. Mod. Phys.\/}}  \bibinfo{volume}{74}
  (\bibinfo{date}{Jan} \bibinfo{year}{2002}), \bibinfo{pages}{47--97}.
\newblock
Issue 1.


\bibitem[\protect\citeauthoryear{Almeida, Guedes, Meira, and Zaki}{Almeida
  et~al\mbox{.}}{2011}]%
        {almeida2011there}
\bibfield{author}{\bibinfo{person}{H{\'e}lio Almeida},
  \bibinfo{person}{Dorgival Guedes}, \bibinfo{person}{Wagner Meira}, {and}
  \bibinfo{person}{Mohammed~J Zaki}.} \bibinfo{year}{2011}\natexlab{}.
\newblock \showarticletitle{Is there a best quality metric for graph
  clusters?}. In \bibinfo{booktitle}{{\em ECML/PKDD}}. Springer,
  \bibinfo{pages}{44--59}.
\newblock


\bibitem[\protect\citeauthoryear{Alon}{Alon}{1997}]%
        {alon1997edge}
\bibfield{author}{\bibinfo{person}{Noga Alon}.}
  \bibinfo{year}{1997}\natexlab{}.
\newblock \showarticletitle{On the edge-expansion of graphs}.
\newblock \bibinfo{journal}{{\em Combinatorics, Probability and Computing\/}}
  \bibinfo{volume}{6}, \bibinfo{number}{2} (\bibinfo{year}{1997}),
  \bibinfo{pages}{145--152}.
\newblock


\bibitem[\protect\citeauthoryear{Arenas, Fernandez, Fortunato, and
  Gomez}{Arenas et~al\mbox{.}}{2008}]%
        {Gomez2008}
\bibfield{author}{\bibinfo{person}{Alex Arenas}, \bibinfo{person}{Alberto
  Fernandez}, \bibinfo{person}{Santo Fortunato}, {and} \bibinfo{person}{Sergio
  Gomez}.} \bibinfo{year}{2008}\natexlab{}.
\newblock \showarticletitle{Motif-based communities in complex networks}.
\newblock \bibinfo{journal}{{\em Journal of Physics A: Mathematical and
  Theoretical\/}} \bibinfo{volume}{41}, \bibinfo{number}{22}
  (\bibinfo{year}{2008}), \bibinfo{pages}{224001}.
\newblock


\bibitem[\protect\citeauthoryear{Arenas, D{\'{i}}az-Guilera, Duch, and
  Alex}{Arenas et~al\mbox{.}}{2005}]%
        {Arenas2005}
\bibfield{author}{\bibinfo{person}{Leon~Danon Arenas}, \bibinfo{person}{Albert
  D{\'{i}}az-Guilera}, \bibinfo{person}{Jordi Duch}, {and}
  \bibinfo{person}{Alex}.} \bibinfo{year}{2005}\natexlab{}.
\newblock \showarticletitle{{Comparing community structure identification}}.
\newblock \bibinfo{journal}{{\em Journal of Statistical Mechanics: Theory and
  Experiment\/}} \bibinfo{number}{09} (\bibinfo{year}{2005}),
  \bibinfo{pages}{P09008}.
\newblock


\bibitem[\protect\citeauthoryear{Balasubramanyan and Cohen}{Balasubramanyan and
  Cohen}{2011}]%
        {Balasubramanyan2011}
\bibfield{author}{\bibinfo{person}{Ramnath Balasubramanyan} {and}
  \bibinfo{person}{William~W Cohen}.} \bibinfo{year}{2011}\natexlab{}.
\newblock \showarticletitle{{Block-LDA: Jointly modeling entity-annotated text
  and entity-entity links}}. In \bibinfo{booktitle}{{\em SDM}}. SIAM,
  \bibinfo{pages}{450--461}.
\newblock


\bibitem[\protect\citeauthoryear{Banerjee, Basu, and Merugu}{Banerjee
  et~al\mbox{.}}{2007}]%
        {banerjee2007multi}
\bibfield{author}{\bibinfo{person}{Arindam Banerjee}, \bibinfo{person}{Sugato
  Basu}, {and} \bibinfo{person}{Srujana Merugu}.}
  \bibinfo{year}{2007}\natexlab{}.
\newblock \showarticletitle{Multi-way clustering on relation graphs}. In
  \bibinfo{booktitle}{{\em SDM}}. SIAM, \bibinfo{pages}{145--156}.
\newblock


\bibitem[\protect\citeauthoryear{Benson, Gleich, and Leskovec}{Benson
  et~al\mbox{.}}{2016}]%
        {benson2016higher}
\bibfield{author}{\bibinfo{person}{Austin~R Benson}, \bibinfo{person}{David~F
  Gleich}, {and} \bibinfo{person}{Jure Leskovec}.}
  \bibinfo{year}{2016}\natexlab{}.
\newblock \showarticletitle{Higher-order organization of complex networks}.
\newblock \bibinfo{journal}{{\em Science\/}} \bibinfo{volume}{353},
  \bibinfo{number}{6295} (\bibinfo{year}{2016}), \bibinfo{pages}{163--166}.
\newblock


\bibitem[\protect\citeauthoryear{Blondel, Guillaume, Lambiotte, and
  Lefebvre}{Blondel et~al\mbox{.}}{2008}]%
        {blondel2008fast}
\bibfield{author}{\bibinfo{person}{Vincent~D Blondel},
  \bibinfo{person}{Jean-Loup Guillaume}, \bibinfo{person}{Renaud Lambiotte},
  {and} \bibinfo{person}{Etienne Lefebvre}.} \bibinfo{year}{2008}\natexlab{}.
\newblock \showarticletitle{Fast unfolding of communities in large networks}.
\newblock \bibinfo{journal}{{\em Journal of statistical mechanics: theory and
  experiment\/}} \bibinfo{volume}{2008}, \bibinfo{number}{10}
  (\bibinfo{year}{2008}), \bibinfo{pages}{P10008}.
\newblock


\bibitem[\protect\citeauthoryear{Boldi, Rosa, Santini, and Vigna}{Boldi
  et~al\mbox{.}}{2011}]%
        {boldi2011layered}
\bibfield{author}{\bibinfo{person}{Paolo Boldi}, \bibinfo{person}{Marco Rosa},
  \bibinfo{person}{Massimo Santini}, {and} \bibinfo{person}{Sebastiano Vigna}.}
  \bibinfo{year}{2011}\natexlab{}.
\newblock \showarticletitle{Layered label propagation: A multiresolution
  coordinate-free ordering for compressing social networks}. In
  \bibinfo{booktitle}{{\em WWW}}. \bibinfo{pages}{587--596}.
\newblock


\bibitem[\protect\citeauthoryear{Boldi and Vigna}{Boldi and Vigna}{2004}]%
        {bvgraph}
\bibfield{author}{\bibinfo{person}{Paolo Boldi} {and}
  \bibinfo{person}{Sebastiano Vigna}.} \bibinfo{year}{2004}\natexlab{}.
\newblock \showarticletitle{The webgraph framework I: compression techniques}.
  In \bibinfo{booktitle}{{\em WWW}}. \bibinfo{pages}{595--602}.
\newblock


\bibitem[\protect\citeauthoryear{Bothorel, Cruz, Magnani, and
  Micenkov{\'{a}}}{Bothorel et~al\mbox{.}}{2015}]%
        {Bothorel2015}
\bibfield{author}{\bibinfo{person}{Cecile Bothorel},
  \bibinfo{person}{Juan~David Cruz}, \bibinfo{person}{Matteo Magnani}, {and}
  \bibinfo{person}{Barbora Micenkov{\'{a}}}.} \bibinfo{year}{2015}\natexlab{}.
\newblock \showarticletitle{{Clustering attributed graphs: Models, measures and
  methods}}.
\newblock \bibinfo{journal}{{\em Network Science\/}} \bibinfo{volume}{3},
  \bibinfo{number}{3} (\bibinfo{year}{2015}), \bibinfo{pages}{408--444}.
\newblock
\showISSN{2050-1242}


\bibitem[\protect\citeauthoryear{Buehrer and Chellapilla}{Buehrer and
  Chellapilla}{2008}]%
        {buehrer2008scalable}
\bibfield{author}{\bibinfo{person}{Gregory Buehrer} {and}
  \bibinfo{person}{Kumar Chellapilla}.} \bibinfo{year}{2008}\natexlab{}.
\newblock \showarticletitle{A scalable pattern mining approach to web graph
  compression with communities}. In \bibinfo{booktitle}{{\em WSDM}}.
  \bibinfo{pages}{95--106}.
\newblock


\bibitem[\protect\citeauthoryear{Cao, Lu, and Xu}{Cao et~al\mbox{.}}{2015}]%
        {grarep}
\bibfield{author}{\bibinfo{person}{Shaosheng Cao}, \bibinfo{person}{Wei Lu},
  {and} \bibinfo{person}{Qiongkai Xu}.} \bibinfo{year}{2015}\natexlab{}.
\newblock \showarticletitle{GraRep: Learning graph representations with global
  structural information}. In \bibinfo{booktitle}{{\em CIKM}}.
  \bibinfo{pages}{891--900}.
\newblock


\bibitem[\protect\citeauthoryear{Chierichetti, Kumar, Lattanzi, Mitzenmacher,
  Panconesi, and Raghavan}{Chierichetti et~al\mbox{.}}{2009}]%
        {chierichetti2009compressing}
\bibfield{author}{\bibinfo{person}{Flavio Chierichetti}, \bibinfo{person}{Ravi
  Kumar}, \bibinfo{person}{Silvio Lattanzi}, \bibinfo{person}{Michael
  Mitzenmacher}, \bibinfo{person}{Alessandro Panconesi}, {and}
  \bibinfo{person}{Prabhakar Raghavan}.} \bibinfo{year}{2009}\natexlab{}.
\newblock \showarticletitle{On compressing social networks}. In
  \bibinfo{booktitle}{{\em KDD}}. \bibinfo{pages}{219--228}.
\newblock


\bibitem[\protect\citeauthoryear{Chung}{Chung}{1996}]%
        {chung1996laplacians}
\bibfield{author}{\bibinfo{person}{Fan~RK Chung}.}
  \bibinfo{year}{1996}\natexlab{}.
\newblock \showarticletitle{Laplacians of graphs and Cheeger's inequalities}.
\newblock \bibinfo{journal}{{\em Combinatorics, Paul Erdos is Eighty\/}}
  \bibinfo{volume}{2}, \bibinfo{number}{157-172} (\bibinfo{year}{1996}),
  \bibinfo{pages}{13--2}.
\newblock


\bibitem[\protect\citeauthoryear{Chung}{Chung}{1997}]%
        {chung1997spectralbook}
\bibfield{author}{\bibinfo{person}{Fan~RK Chung}.}
  \bibinfo{year}{1997}\natexlab{}.
\newblock \bibinfo{booktitle}{{\em Spectral graph theory}}.
\newblock Number~92. \bibinfo{publisher}{Amer. Math. Soc.}
\newblock


\bibitem[\protect\citeauthoryear{Chung and Oden}{Chung and Oden}{2000}]%
        {chung2000weighted}
\bibfield{author}{\bibinfo{person}{Fan~RK Chung} {and} \bibinfo{person}{Kevin
  Oden}.} \bibinfo{year}{2000}\natexlab{}.
\newblock \showarticletitle{Weighted graph Laplacians and isoperimetric
  inequalities}.
\newblock \bibinfo{journal}{{\it Pacific J. Math.}} \bibinfo{volume}{192},
  \bibinfo{number}{2} (\bibinfo{year}{2000}), \bibinfo{pages}{257--273}.
\newblock


\bibitem[\protect\citeauthoryear{Combe, Largeron, Egyed-Zsigmond, and
  G{\'{e}}ry}{Combe et~al\mbox{.}}{2012}]%
        {Combe2012}
\bibfield{author}{\bibinfo{person}{D Combe}, \bibinfo{person}{C Largeron},
  \bibinfo{person}{E Egyed-Zsigmond}, {and} \bibinfo{person}{M G{\'{e}}ry}.}
  \bibinfo{year}{2012}\natexlab{}.
\newblock \showarticletitle{{Combining Relations and Text in Scientific Network
  Clustering}}. In \bibinfo{booktitle}{{\em ASONAM}}.
  \bibinfo{pages}{1248--1253}.
\newblock
\showISBNx{VO -}


\bibitem[\protect\citeauthoryear{Cook}{Cook}{1971}]%
        {Cook1971}
\bibfield{author}{\bibinfo{person}{Stephen~A. Cook}.}
  \bibinfo{year}{1971}\natexlab{}.
\newblock \showarticletitle{The Complexity of Theorem-proving Procedures}. In
  \bibinfo{booktitle}{{\em STOC}}. \bibinfo{publisher}{ACM},
  \bibinfo{address}{New York, NY, USA}, \bibinfo{pages}{151--158}.
\newblock


\bibitem[\protect\citeauthoryear{Coscia, Giannotti, and Pedreschi}{Coscia
  et~al\mbox{.}}{2011}]%
        {Coscia2011}
\bibfield{author}{\bibinfo{person}{Michele Coscia}, \bibinfo{person}{Fosca
  Giannotti}, {and} \bibinfo{person}{Dino Pedreschi}.}
  \bibinfo{year}{2011}\natexlab{}.
\newblock \showarticletitle{{A classification for community discovery methods
  in complex networks}}.
\newblock \bibinfo{journal}{{\em Statistical Analysis and Data Mining\/}}
  \bibinfo{volume}{4}, \bibinfo{number}{5} (\bibinfo{date}{sep}
  \bibinfo{year}{2011}), \bibinfo{pages}{512--546}.
\newblock
\showISSN{1932-1864}


\bibitem[\protect\citeauthoryear{Dhillon}{Dhillon}{2001}]%
        {dhillon2001co}
\bibfield{author}{\bibinfo{person}{I.S. Dhillon}.}
  \bibinfo{year}{2001}\natexlab{}.
\newblock \showarticletitle{{Co-clustering documents and words using bipartite
  spectral graph partitioning}}. In \bibinfo{booktitle}{{\em SIGKDD}}.
  \bibinfo{pages}{269--274}.
\newblock


\bibitem[\protect\citeauthoryear{Erd{\H{o}}s and Hajnal}{Erd{\H{o}}s and
  Hajnal}{1966}]%
        {erdHos1966chromatic}
\bibfield{author}{\bibinfo{person}{Paul Erd{\H{o}}s} {and}
  \bibinfo{person}{Andr{\'a}s Hajnal}.} \bibinfo{year}{1966}\natexlab{}.
\newblock \showarticletitle{On chromatic number of graphs and set-systems}.
\newblock \bibinfo{journal}{{\em Acta Mathematica Academiae Scientiarum
  Hungarica\/}} \bibinfo{volume}{17}, \bibinfo{number}{1-2}
  (\bibinfo{year}{1966}), \bibinfo{pages}{61--99}.
\newblock


\bibitem[\protect\citeauthoryear{Felzenszwalb and Huttenlocher}{Felzenszwalb
  and Huttenlocher}{2004}]%
        {felzenszwalb2004efficient}
\bibfield{author}{\bibinfo{person}{Pedro~F Felzenszwalb} {and}
  \bibinfo{person}{Daniel~P Huttenlocher}.} \bibinfo{year}{2004}\natexlab{}.
\newblock \showarticletitle{Efficient graph-based image segmentation}.
\newblock \bibinfo{journal}{{\em IJCV\/}} \bibinfo{volume}{59},
  \bibinfo{number}{2} (\bibinfo{year}{2004}), \bibinfo{pages}{167--181}.
\newblock


\bibitem[\protect\citeauthoryear{Fortunato}{Fortunato}{2010}]%
        {Fortunato2010}
\bibfield{author}{\bibinfo{person}{Santo Fortunato}.}
  \bibinfo{year}{2010}\natexlab{}.
\newblock \showarticletitle{{Community detection in graphs}}.
\newblock \bibinfo{journal}{{\em Physics Reports\/}} \bibinfo{volume}{486},
  \bibinfo{number}{3} (\bibinfo{year}{2010}), \bibinfo{pages}{75--174}.
\newblock
\showISSN{0370-1573}


\bibitem[\protect\citeauthoryear{Fortunato and Barth{\'{e}}lemy}{Fortunato and
  Barth{\'{e}}lemy}{2007}]%
        {Fortunato2007}
\bibfield{author}{\bibinfo{person}{Santo Fortunato} {and} \bibinfo{person}{Marc
  Barth{\'{e}}lemy}.} \bibinfo{year}{2007}\natexlab{}.
\newblock \showarticletitle{{Resolution limit in community detection}}.
\newblock \bibinfo{journal}{{\em PNAS\/}} \bibinfo{volume}{104},
  \bibinfo{number}{1} (\bibinfo{year}{2007}), \bibinfo{pages}{36--41}.
\newblock
\showISSN{0027-8424}


\bibitem[\protect\citeauthoryear{Fortunato and Hric}{Fortunato and
  Hric}{2016}]%
        {Fortunato2016}
\bibfield{author}{\bibinfo{person}{Santo Fortunato} {and}
  \bibinfo{person}{Darko Hric}.} \bibinfo{year}{2016}\natexlab{}.
\newblock \showarticletitle{{Community detection in networks: A user guide}}.
\newblock \bibinfo{journal}{{\em Physics Reports\/}}  \bibinfo{volume}{659}
  (\bibinfo{year}{2016}), \bibinfo{pages}{1--44}.
\newblock
\showISSN{0370-1573}


\bibitem[\protect\citeauthoryear{Gaertler}{Gaertler}{2005}]%
        {gaertler2005clustering}
\bibfield{author}{\bibinfo{person}{Marco Gaertler}.}
  \bibinfo{year}{2005}\natexlab{}.
\newblock \showarticletitle{Clustering}.
\newblock In \bibinfo{booktitle}{{\em Network analysis}}.
  \bibinfo{publisher}{Springer}, \bibinfo{pages}{178--215}.
\newblock


\bibitem[\protect\citeauthoryear{Ge, Ester, Gao, Hu, Bhattacharya, and
  Ben-Moshe}{Ge et~al\mbox{.}}{2008}]%
        {Ge2008}
\bibfield{author}{\bibinfo{person}{Rong Ge}, \bibinfo{person}{Martin Ester},
  \bibinfo{person}{Byron~J Gao}, \bibinfo{person}{Zengjian Hu},
  \bibinfo{person}{Binay Bhattacharya}, {and} \bibinfo{person}{Boaz
  Ben-Moshe}.} \bibinfo{year}{2008}\natexlab{}.
\newblock \showarticletitle{{Joint Cluster Analysis of Attribute Data and
  Relationship Data: The Connected K-center Problem, Algorithms and
  Applications}}.
\newblock \bibinfo{journal}{{\em TKDD\/}} \bibinfo{volume}{2},
  \bibinfo{number}{2} (\bibinfo{date}{jul} \bibinfo{year}{2008}),
  \bibinfo{pages}{7:1----7:35}.
\newblock
\showISSN{1556-4681}


\bibitem[\protect\citeauthoryear{Gleich and Seshadhri}{Gleich and
  Seshadhri}{2012}]%
        {gleich2012vertex}
\bibfield{author}{\bibinfo{person}{David~F Gleich} {and} \bibinfo{person}{C
  Seshadhri}.} \bibinfo{year}{2012}\natexlab{}.
\newblock \showarticletitle{Vertex neighborhoods, low conductance cuts, and
  good seeds for local community methods}. In \bibinfo{booktitle}{{\em
  SIGKDD}}. \bibinfo{pages}{597--605}.
\newblock


\bibitem[\protect\citeauthoryear{Golub and {Van Loan}}{Golub and {Van
  Loan}}{2012}]%
        {Golub2012}
\bibfield{author}{\bibinfo{person}{Gene~H Golub} {and}
  \bibinfo{person}{Charles~F {Van Loan}}.} \bibinfo{year}{2012}\natexlab{}.
\newblock \bibinfo{booktitle}{{\em {Matrix computations}}}.
  Vol.~\bibinfo{volume}{3}.
\newblock \bibinfo{publisher}{JHU Press}.
\newblock


\bibitem[\protect\citeauthoryear{Grover and Leskovec}{Grover and
  Leskovec}{2016}]%
        {node2vec}
\bibfield{author}{\bibinfo{person}{Aditya Grover} {and} \bibinfo{person}{Jure
  Leskovec}.} \bibinfo{year}{2016}\natexlab{}.
\newblock \showarticletitle{node2vec: Scalable feature learning for networks}.
  In \bibinfo{booktitle}{{\em KDD}}. \bibinfo{pages}{855--864}.
\newblock


\bibitem[\protect\citeauthoryear{Harenberg, Bello, Gjeltema, Ranshous,
  Harlalka, Seay, Padmanabhan, and Samatova}{Harenberg et~al\mbox{.}}{2014}]%
        {Harenberg2014}
\bibfield{author}{\bibinfo{person}{Steve Harenberg}, \bibinfo{person}{Gonzalo
  Bello}, \bibinfo{person}{L Gjeltema}, \bibinfo{person}{Stephen Ranshous},
  \bibinfo{person}{Jitendra Harlalka}, \bibinfo{person}{Ramona Seay},
  \bibinfo{person}{Kanchana Padmanabhan}, {and} \bibinfo{person}{Nagiza
  Samatova}.} \bibinfo{year}{2014}\natexlab{}.
\newblock \showarticletitle{{Community detection in large-scale networks: a
  survey and empirical evaluation}}.
\newblock \bibinfo{journal}{{\em Wiley Interdisciplinary Reviews: Computational
  Statistics\/}} \bibinfo{volume}{6}, \bibinfo{number}{6}
  (\bibinfo{year}{2014}), \bibinfo{pages}{426--439}.
\newblock


\bibitem[\protect\citeauthoryear{Hendrickson and Leland}{Hendrickson and
  Leland}{1995}]%
        {hendrickson1995improved}
\bibfield{author}{\bibinfo{person}{Bruce Hendrickson} {and}
  \bibinfo{person}{Robert Leland}.} \bibinfo{year}{1995}\natexlab{}.
\newblock \showarticletitle{An improved spectral graph partitioning algorithm
  for mapping parallel computations}.
\newblock \bibinfo{journal}{{\em SIAM Journal on Scientific Computing\/}}
  \bibinfo{volume}{16}, \bibinfo{number}{2} (\bibinfo{year}{1995}),
  \bibinfo{pages}{452--469}.
\newblock


\bibitem[\protect\citeauthoryear{Hoory, Linial, and Wigderson}{Hoory
  et~al\mbox{.}}{2006}]%
        {hoory2006expander}
\bibfield{author}{\bibinfo{person}{Shlomo Hoory}, \bibinfo{person}{Nathan
  Linial}, {and} \bibinfo{person}{Avi Wigderson}.}
  \bibinfo{year}{2006}\natexlab{}.
\newblock \showarticletitle{Expander graphs and their applications}.
\newblock \bibinfo{journal}{{\it Bull. Amer. Math. Soc.}} \bibinfo{volume}{43},
  \bibinfo{number}{4} (\bibinfo{year}{2006}), \bibinfo{pages}{439--561}.
\newblock


\bibitem[\protect\citeauthoryear{Kannan, Vempala, and Vetta}{Kannan
  et~al\mbox{.}}{2004}]%
        {Kannan2004}
\bibfield{author}{\bibinfo{person}{Ravi Kannan}, \bibinfo{person}{Santosh
  Vempala}, {and} \bibinfo{person}{Adrian Vetta}.}
  \bibinfo{year}{2004}\natexlab{}.
\newblock \showarticletitle{{On Clusterings: Good, Bad and Spectral}}.
\newblock \bibinfo{journal}{{\em J. ACM\/}} \bibinfo{volume}{51},
  \bibinfo{number}{3} (\bibinfo{date}{May} \bibinfo{year}{2004}),
  \bibinfo{pages}{497--515}.
\newblock
\showISSN{0004-5411}


\bibitem[\protect\citeauthoryear{Karande, Chellapilla, and Andersen}{Karande
  et~al\mbox{.}}{2009}]%
        {karande2009speeding}
\bibfield{author}{\bibinfo{person}{Chinmay Karande}, \bibinfo{person}{Kumar
  Chellapilla}, {and} \bibinfo{person}{Reid Andersen}.}
  \bibinfo{year}{2009}\natexlab{}.
\newblock \showarticletitle{Speeding up algorithms on compressed web graphs}.
\newblock \bibinfo{journal}{{\em Internet Mathematics\/}} \bibinfo{volume}{6},
  \bibinfo{number}{3} (\bibinfo{year}{2009}), \bibinfo{pages}{373--398}.
\newblock


\bibitem[\protect\citeauthoryear{Khuller and Saha}{Khuller and Saha}{2009}]%
        {khuller2009finding}
\bibfield{author}{\bibinfo{person}{Samir Khuller} {and} \bibinfo{person}{Barna
  Saha}.} \bibinfo{year}{2009}\natexlab{}.
\newblock \showarticletitle{On finding dense subgraphs}. In
  \bibinfo{booktitle}{{\em International Colloquium on Automata, Languages, and
  Programming}}. Springer, \bibinfo{pages}{597--608}.
\newblock


\bibitem[\protect\citeauthoryear{Knabe, Nehaniv, and Schilstra}{Knabe
  et~al\mbox{.}}{2008}]%
        {Knabe2008}
\bibfield{author}{\bibinfo{person}{Johannes~F Knabe},
  \bibinfo{person}{Chrystopher~L Nehaniv}, {and} \bibinfo{person}{Maria~J
  Schilstra}.} \bibinfo{year}{2008}\natexlab{}.
\newblock \showarticletitle{{Do motifs reflect evolved function? No convergent
  evolution of genetic regulatory network subgraph topologies}}.
\newblock \bibinfo{journal}{{\em Biosystems\/}} \bibinfo{volume}{94},
  \bibinfo{number}{1} (\bibinfo{year}{2008}), \bibinfo{pages}{68--74}.
\newblock
\showISSN{0303-2647}


\bibitem[\protect\citeauthoryear{Lancichinetti and Fortunato}{Lancichinetti and
  Fortunato}{2009}]%
        {Lancichinetti2009}
\bibfield{author}{\bibinfo{person}{Andrea Lancichinetti} {and}
  \bibinfo{person}{Santo Fortunato}.} \bibinfo{year}{2009}\natexlab{}.
\newblock \showarticletitle{{Community detection algorithms: A comparative
  analysis}}.
\newblock \bibinfo{journal}{{\em Physical Review E\/}} \bibinfo{volume}{80},
  \bibinfo{number}{5} (\bibinfo{date}{nov} \bibinfo{year}{2009}),
  \bibinfo{pages}{56117}.
\newblock


\bibitem[\protect\citeauthoryear{Leighton and Rao}{Leighton and Rao}{1999}]%
        {Leighton1999}
\bibfield{author}{\bibinfo{person}{Tom Leighton} {and} \bibinfo{person}{Satish
  Rao}.} \bibinfo{year}{1999}\natexlab{}.
\newblock \showarticletitle{{Multicommodity Max-flow Min-cut Theorems and Their
  Use in Designing Approximation Algorithms}}.
\newblock \bibinfo{journal}{{\em J. ACM\/}} \bibinfo{volume}{46},
  \bibinfo{number}{6} (\bibinfo{date}{nov} \bibinfo{year}{1999}),
  \bibinfo{pages}{787--832}.
\newblock
\showISSN{0004-5411}


\bibitem[\protect\citeauthoryear{Leskovec, Lang, and Mahoney}{Leskovec
  et~al\mbox{.}}{2010}]%
        {Leskovec2010}
\bibfield{author}{\bibinfo{person}{Jure Leskovec}, \bibinfo{person}{Kevin~J
  Lang}, {and} \bibinfo{person}{Michael Mahoney}.}
  \bibinfo{year}{2010}\natexlab{}.
\newblock \showarticletitle{{Empirical Comparison of Algorithms for Network
  Community Detection}}. In \bibinfo{booktitle}{{\em WWW}}.
  \bibinfo{address}{New York, NY, USA}, \bibinfo{pages}{631--640}.
\newblock
\showISBNx{978-1-60558-799-8}


\bibitem[\protect\citeauthoryear{Liakos, Papakonstantinopoulou, and
  Sioutis}{Liakos et~al\mbox{.}}{2014}]%
        {liakos2014pushing}
\bibfield{author}{\bibinfo{person}{Panagiotis Liakos}, \bibinfo{person}{Katia
  Papakonstantinopoulou}, {and} \bibinfo{person}{Michael Sioutis}.}
  \bibinfo{year}{2014}\natexlab{}.
\newblock \showarticletitle{Pushing the envelope in graph compression}. In
  \bibinfo{booktitle}{{\em CIKM}}. \bibinfo{pages}{1549--1558}.
\newblock


\bibitem[\protect\citeauthoryear{Liben-Nowell and Kleinberg}{Liben-Nowell and
  Kleinberg}{2003}]%
        {linkpred:liben07}
\bibfield{author}{\bibinfo{person}{David Liben-Nowell} {and}
  \bibinfo{person}{Jon Kleinberg}.} \bibinfo{year}{2003}\natexlab{}.
\newblock \showarticletitle{The link prediction problem for social networks}.
  In \bibinfo{booktitle}{{\em CIKM}}. \bibinfo{pages}{556--559}.
\newblock


\bibitem[\protect\citeauthoryear{Liu, Niculescu-Mizil, and Gryc}{Liu
  et~al\mbox{.}}{2009}]%
        {Liu2009}
\bibfield{author}{\bibinfo{person}{Yan Liu}, \bibinfo{person}{Alexandru
  Niculescu-Mizil}, {and} \bibinfo{person}{Wojciech Gryc}.}
  \bibinfo{year}{2009}\natexlab{}.
\newblock \showarticletitle{{Topic-link LDA: Joint Models of Topic and Author
  Community}}. In \bibinfo{booktitle}{{\em ICML}}. \bibinfo{address}{New York,
  NY, USA}, \bibinfo{pages}{665--672}.
\newblock
\showISBNx{978-1-60558-516-1}


\bibitem[\protect\citeauthoryear{Malliaros and Vazirgiannis}{Malliaros and
  Vazirgiannis}{2013}]%
        {Malliaros2013}
\bibfield{author}{\bibinfo{person}{Fragkiskos~D Malliaros} {and}
  \bibinfo{person}{Michalis Vazirgiannis}.} \bibinfo{year}{2013}\natexlab{}.
\newblock \showarticletitle{{Clustering and community detection in directed
  networks: A survey}}.
\newblock \bibinfo{journal}{{\em Physics Reports\/}} \bibinfo{volume}{533},
  \bibinfo{number}{4} (\bibinfo{year}{2013}), \bibinfo{pages}{95--142}.
\newblock
\showISSN{0370-1573}


\bibitem[\protect\citeauthoryear{Marcus and Shavitt}{Marcus and
  Shavitt}{2012}]%
        {rage}
\bibfield{author}{\bibinfo{person}{Dror Marcus} {and} \bibinfo{person}{Yuval
  Shavitt}.} \bibinfo{year}{2012}\natexlab{}.
\newblock \showarticletitle{RAGE--a rapid graphlet enumerator for large
  networks}.
\newblock \bibinfo{journal}{{\em Computer Networks\/}} \bibinfo{volume}{56},
  \bibinfo{number}{2} (\bibinfo{year}{2012}), \bibinfo{pages}{810--819}.
\newblock


\bibitem[\protect\citeauthoryear{Milo, Shen-Orr, Itzkovitz, Kashtan,
  Chklovskii, and Alon}{Milo et~al\mbox{.}}{2002}]%
        {Milo2002}
\bibfield{author}{\bibinfo{person}{R Milo}, \bibinfo{person}{S Shen-Orr},
  \bibinfo{person}{S Itzkovitz}, \bibinfo{person}{N Kashtan},
  \bibinfo{person}{D Chklovskii}, {and} \bibinfo{person}{U Alon}.}
  \bibinfo{year}{2002}\natexlab{}.
\newblock \showarticletitle{{Network Motifs: Simple Building Blocks of Complex
  Networks}}.
\newblock \bibinfo{journal}{{\em Science\/}} \bibinfo{volume}{298},
  \bibinfo{number}{5594} (\bibinfo{year}{2002}), \bibinfo{pages}{824--827}.
\newblock
\showISSN{0036-8075}


\bibitem[\protect\citeauthoryear{Neville, Adler, and Jensen}{Neville
  et~al\mbox{.}}{2003}]%
        {Neville2003}
\bibfield{author}{\bibinfo{person}{Jennifer Neville}, \bibinfo{person}{Micah
  Adler}, {and} \bibinfo{person}{David Jensen}.}
  \bibinfo{year}{2003}\natexlab{}.
\newblock \showarticletitle{Clustering relational data using attribute and link
  information}. In \bibinfo{booktitle}{{\em IJCAI Workshop}}.
  \bibinfo{pages}{9--15}.
\newblock


\bibitem[\protect\citeauthoryear{Newman}{Newman}{2011}]%
        {Newman2011}
\bibfield{author}{\bibinfo{person}{M~E~J Newman}.}
  \bibinfo{year}{2011}\natexlab{}.
\newblock \showarticletitle{{Communities, modules and large-scale structure in
  networks}}.
\newblock \bibinfo{journal}{{\em Nature Physics\/}}  \bibinfo{volume}{8}
  (\bibinfo{date}{dec} \bibinfo{year}{2011}), \bibinfo{pages}{25}.
\newblock


\bibitem[\protect\citeauthoryear{Newman and Girvan}{Newman and Girvan}{2004}]%
        {Newman2004}
\bibfield{author}{\bibinfo{person}{M~E~J Newman} {and} \bibinfo{person}{M
  Girvan}.} \bibinfo{year}{2004}\natexlab{}.
\newblock \showarticletitle{{Finding and evaluating community structure in
  networks}}.
\newblock \bibinfo{journal}{{\em Physical Review E\/}} \bibinfo{volume}{69},
  \bibinfo{number}{2} (\bibinfo{date}{feb} \bibinfo{year}{2004}),
  \bibinfo{pages}{26113}.
\newblock


\bibitem[\protect\citeauthoryear{Ng, Jordan, and Weiss}{Ng
  et~al\mbox{.}}{2002}]%
        {Ng2002}
\bibfield{author}{\bibinfo{person}{Andrew~Y Ng}, \bibinfo{person}{Michael~I
  Jordan}, {and} \bibinfo{person}{Yair Weiss}.}
  \bibinfo{year}{2002}\natexlab{}.
\newblock \showarticletitle{{On spectral clustering: Analysis and an
  algorithm}}. In \bibinfo{booktitle}{{\em NIPS}}. \bibinfo{pages}{849--856}.
\newblock


\bibitem[\protect\citeauthoryear{Paranjape, Benson, and Leskovec}{Paranjape
  et~al\mbox{.}}{2017}]%
        {Paranjape2017}
\bibfield{author}{\bibinfo{person}{Ashwin Paranjape}, \bibinfo{person}{Austin~R
  Benson}, {and} \bibinfo{person}{Jure Leskovec}.}
  \bibinfo{year}{2017}\natexlab{}.
\newblock \showarticletitle{{Motifs in Temporal Networks}}. In
  \bibinfo{booktitle}{{\em WSDM}}. \bibinfo{pages}{601--610}.
\newblock
\showISBNx{978-1-4503-4675-7}


\bibitem[\protect\citeauthoryear{Perozzi, Al-Rfou, and Skiena}{Perozzi
  et~al\mbox{.}}{2014}]%
        {deepwalk}
\bibfield{author}{\bibinfo{person}{Bryan Perozzi}, \bibinfo{person}{Rami
  Al-Rfou}, {and} \bibinfo{person}{Steven Skiena}.}
  \bibinfo{year}{2014}\natexlab{}.
\newblock \showarticletitle{Deepwalk: Online learning of social
  representations}. In \bibinfo{booktitle}{{\em KDD}}.
  \bibinfo{pages}{701--710}.
\newblock


\bibitem[\protect\citeauthoryear{Pr{\v{z}}ulj, Corneil, and
  Jurisica}{Pr{\v{z}}ulj et~al\mbox{.}}{2004}]%
        {Przulj2004}
\bibfield{author}{\bibinfo{person}{N Pr{\v{z}}ulj}, \bibinfo{person}{D~G
  Corneil}, {and} \bibinfo{person}{I Jurisica}.}
  \bibinfo{year}{2004}\natexlab{}.
\newblock \showarticletitle{{Modeling interactome: scale-free or geometric?}}
\newblock \bibinfo{journal}{{\em Bioinformatics\/}} \bibinfo{volume}{20},
  \bibinfo{number}{18} (\bibinfo{date}{dec} \bibinfo{year}{2004}),
  \bibinfo{pages}{3508--3515}.
\newblock
\showISSN{1367-4803}


\bibitem[\protect\citeauthoryear{Raghavan, Albert, and Kumara}{Raghavan
  et~al\mbox{.}}{2007}]%
        {raghavan2007near}
\bibfield{author}{\bibinfo{person}{Usha~Nandini Raghavan},
  \bibinfo{person}{R{\'e}ka Albert}, {and} \bibinfo{person}{Soundar Kumara}.}
  \bibinfo{year}{2007}\natexlab{}.
\newblock \showarticletitle{Near linear time algorithm to detect community
  structures in large-scale networks}.
\newblock \bibinfo{journal}{{\em Physical review E\/}} \bibinfo{volume}{76},
  \bibinfo{number}{3} (\bibinfo{year}{2007}), \bibinfo{pages}{036106}.
\newblock


\bibitem[\protect\citeauthoryear{Rossi and Ahmed}{Rossi and Ahmed}{2014}]%
        {rossi14coloring-networks}
\bibfield{author}{\bibinfo{person}{Ryan~A. Rossi} {and}
  \bibinfo{person}{Nesreen~K. Ahmed}.} \bibinfo{year}{2014}\natexlab{}.
\newblock \showarticletitle{Coloring Large Complex Networks}.
\newblock \bibinfo{journal}{{\em Social Network Analysis and Mining\/}}
  \bibinfo{volume}{4}, \bibinfo{number}{1}, Article \bibinfo{articleno}{228}
  (\bibinfo{year}{2014}), \bibinfo{numpages}{37}~pages.
\newblock
\showISSN{1869-5450}


\bibitem[\protect\citeauthoryear{Rossi and Ahmed}{Rossi and Ahmed}{2015a}]%
        {nr}
\bibfield{author}{\bibinfo{person}{Ryan~A. Rossi} {and}
  \bibinfo{person}{Nesreen~K. Ahmed}.} \bibinfo{year}{2015}\natexlab{a}.
\newblock \showarticletitle{The Network Data Repository with Interactive Graph
  Analytics and Visualization}. In \bibinfo{booktitle}{{\em AAAI}}.
\newblock
\showURL{%
\url{http://networkrepository.com}}


\bibitem[\protect\citeauthoryear{Rossi and Ahmed}{Rossi and Ahmed}{2015b}]%
        {Rossi2015}
\bibfield{author}{\bibinfo{person}{R~A Rossi} {and} \bibinfo{person}{N~K
  Ahmed}.} \bibinfo{year}{2015}\natexlab{b}.
\newblock \showarticletitle{{Role Discovery in Networks}}.
\newblock \bibinfo{journal}{{\em TKDE\/}} \bibinfo{volume}{27},
  \bibinfo{number}{4} (\bibinfo{year}{2015}), \bibinfo{pages}{1112--1131}.
\newblock
\showISSN{1041-4347 VO - 27}


\bibitem[\protect\citeauthoryear{Rossi, Ahmed, and Koh}{Rossi
  et~al\mbox{.}}{2018a}]%
        {HONE}
\bibfield{author}{\bibinfo{person}{Ryan~A. Rossi}, \bibinfo{person}{Nesreen~K.
  Ahmed}, {and} \bibinfo{person}{Eunyee Koh}.}
  \bibinfo{year}{2018}\natexlab{a}.
\newblock \showarticletitle{Higher-Order Network Representation Learning}. In
  \bibinfo{booktitle}{{\em WWW}}.
\newblock


\bibitem[\protect\citeauthoryear{Rossi, Ahmed, Koh, Kim, Rao, and
  Abbasi-Yadkori}{Rossi et~al\mbox{.}}{2018b}]%
        {Rossi2018a}
\bibfield{author}{\bibinfo{person}{Ryan~A. Rossi}, \bibinfo{person}{Nesreen~K.
  Ahmed}, \bibinfo{person}{Eunyee Koh}, \bibinfo{person}{Sungchul Kim},
  \bibinfo{person}{Anup Rao}, {and} \bibinfo{person}{Yasin Abbasi-Yadkori}.}
  \bibinfo{year}{2018}\natexlab{b}.
\newblock \showarticletitle{{HONE: Higher-Order Network Embeddings}}.
\newblock \bibinfo{journal}{{\em arXiv:1801.09303\/}} (\bibinfo{year}{2018}).
\newblock


\bibitem[\protect\citeauthoryear{Rossi, Gleich, and Gebremedhin}{Rossi
  et~al\mbox{.}}{2015}]%
        {rossi2015pmc-sisc}
\bibfield{author}{\bibinfo{person}{Ryan~A. Rossi}, \bibinfo{person}{David~F.
  Gleich}, {and} \bibinfo{person}{Assefaw~H. Gebremedhin}.}
  \bibinfo{year}{2015}\natexlab{}.
\newblock \showarticletitle{Parallel Maximum Clique Algorithms with
  Applications to Network Analysis}.
\newblock \bibinfo{journal}{{\em SISC\/}} \bibinfo{volume}{37},
  \bibinfo{number}{5} (\bibinfo{year}{2015}), \bibinfo{pages}{28}.
\newblock


\bibitem[\protect\citeauthoryear{Rossi and Zhou}{Rossi and Zhou}{2015}]%
        {pcmf-dsaa}
\bibfield{author}{\bibinfo{person}{Ryan~A. Rossi} {and} \bibinfo{person}{Rong
  Zhou}.} \bibinfo{year}{2015}\natexlab{}.
\newblock \showarticletitle{Scalable Relational Learning for Large
  Heterogeneous Networks}. In \bibinfo{booktitle}{{\em DSAA}}.
  \bibinfo{pages}{1--10}.
\newblock


\bibitem[\protect\citeauthoryear{Rossi and Zhou}{Rossi and Zhou}{2016}]%
        {rossi16collective-factor}
\bibfield{author}{\bibinfo{person}{Ryan~A. Rossi} {and} \bibinfo{person}{Rong
  Zhou}.} \bibinfo{year}{2016}\natexlab{}.
\newblock \showarticletitle{Parallel Collective Factorization for Modeling
  Large Heterogeneous Networks}. In \bibinfo{booktitle}{{\em SNAM}}.
  \bibinfo{pages}{30}.
\newblock


\bibitem[\protect\citeauthoryear{Rossi and Zhou}{Rossi and Zhou}{2018}]%
        {rossi2018compressing-graphs-cliques}
\bibfield{author}{\bibinfo{person}{Ryan~A. Rossi} {and} \bibinfo{person}{Rong
  Zhou}.} \bibinfo{year}{2018}\natexlab{}.
\newblock \showarticletitle{GraphZIP: A Clique-based Sparse Graph Compression
  Method}.
\newblock \bibinfo{journal}{{\em Journal of Big Data\/}} \bibinfo{volume}{5},
  \bibinfo{number}{1} (\bibinfo{year}{2018}), \bibinfo{pages}{14}.
\newblock


\bibitem[\protect\citeauthoryear{Rossi, Zhou, and Ahmed}{Rossi
  et~al\mbox{.}}{2017}]%
        {rossi17graphlet-est}
\bibfield{author}{\bibinfo{person}{Ryan~A. Rossi}, \bibinfo{person}{Rong Zhou},
  {and} \bibinfo{person}{Nesreen~K. Ahmed}.} \bibinfo{year}{2017}\natexlab{}.
\newblock \showarticletitle{Estimation of Graphlet Statistics}. In
  \bibinfo{booktitle}{{\em arXiv:1701.01772v1}}. \bibinfo{pages}{1--14}.
\newblock


\bibitem[\protect\citeauthoryear{Rossi, Zhou, and Ahmed}{Rossi
  et~al\mbox{.}}{2018a}]%
        {deepGL}
\bibfield{author}{\bibinfo{person}{Ryan~A. Rossi}, \bibinfo{person}{Rong Zhou},
  {and} \bibinfo{person}{Nesreen~K. Ahmed}.} \bibinfo{year}{2018}\natexlab{a}.
\newblock \showarticletitle{Deep Inductive Network Representation Learning}. In
  \bibinfo{booktitle}{{\em WWW BigNet}}. \bibinfo{pages}{8}.
\newblock


\bibitem[\protect\citeauthoryear{Rossi, Zhou, and Ahmed}{Rossi
  et~al\mbox{.}}{2018b}]%
        {Rossi2018b}
\bibfield{author}{\bibinfo{person}{Ryan~A. Rossi}, \bibinfo{person}{Rong Zhou},
  {and} \bibinfo{person}{Nesreen~K. Ahmed}.} \bibinfo{year}{2018}\natexlab{b}.
\newblock \showarticletitle{Estimation of Graphlet Counts in Massive Networks}.
  In \bibinfo{booktitle}{{\em TNNLS}}. \bibinfo{pages}{1--14}.
\newblock


\bibitem[\protect\citeauthoryear{Schaeffer}{Schaeffer}{2007}]%
        {Schaeffer2007}
\bibfield{author}{\bibinfo{person}{Satu~Elisa Schaeffer}.}
  \bibinfo{year}{2007}\natexlab{}.
\newblock \showarticletitle{{Graph clustering}}.
\newblock \bibinfo{journal}{{\em Computer Science Review\/}}
  \bibinfo{volume}{1}, \bibinfo{number}{1} (\bibinfo{year}{2007}),
  \bibinfo{pages}{27--64}.
\newblock
\showISSN{1574-0137}


\bibitem[\protect\citeauthoryear{Shen-Orr, Milo, Mangan, and Alon}{Shen-Orr
  et~al\mbox{.}}{2002}]%
        {Shen-Orr2002}
\bibfield{author}{\bibinfo{person}{Shai~S Shen-Orr}, \bibinfo{person}{Ron
  Milo}, \bibinfo{person}{Shmoolik Mangan}, {and} \bibinfo{person}{Uri Alon}.}
  \bibinfo{year}{2002}\natexlab{}.
\newblock \showarticletitle{{Network motifs in the transcriptional regulation
  network of Escherichia coli}}.
\newblock \bibinfo{journal}{{\em Nature Genetics\/}}  \bibinfo{volume}{31}
  (\bibinfo{date}{Apr} \bibinfo{year}{2002}), \bibinfo{pages}{64}.
\newblock


\bibitem[\protect\citeauthoryear{Shi, Li, Zhang, Sun, and Yu}{Shi
  et~al\mbox{.}}{2017}]%
        {Shi2017}
\bibfield{author}{\bibinfo{person}{C Shi}, \bibinfo{person}{Y Li},
  \bibinfo{person}{J Zhang}, \bibinfo{person}{Y Sun}, {and}
  \bibinfo{person}{P~S Yu}.} \bibinfo{year}{2017}\natexlab{}.
\newblock \showarticletitle{{A Survey of Heterogeneous Information Network
  Analysis}}.
\newblock \bibinfo{journal}{{\em TKDE\/}} \bibinfo{volume}{29},
  \bibinfo{number}{1} (\bibinfo{year}{2017}), \bibinfo{pages}{17--37}.
\newblock
\showISSN{1041-4347 VO - 29}


\bibitem[\protect\citeauthoryear{Shi and Malik}{Shi and Malik}{2000}]%
        {shi2000normalized}
\bibfield{author}{\bibinfo{person}{Jianbo Shi} {and} \bibinfo{person}{Jitendra
  Malik}.} \bibinfo{year}{2000}\natexlab{}.
\newblock \showarticletitle{Normalized cuts and image segmentation}.
\newblock \bibinfo{journal}{{\em TPAMI\/}} \bibinfo{volume}{22},
  \bibinfo{number}{8} (\bibinfo{year}{2000}), \bibinfo{pages}{888--905}.
\newblock


\bibitem[\protect\citeauthoryear{Shin, Eliassi-Rad, and Faloutsos}{Shin
  et~al\mbox{.}}{2016}]%
        {shin2016corescope}
\bibfield{author}{\bibinfo{person}{Kijung Shin}, \bibinfo{person}{Tina
  Eliassi-Rad}, {and} \bibinfo{person}{Christos Faloutsos}.}
  \bibinfo{year}{2016}\natexlab{}.
\newblock \showarticletitle{CoreScope: Graph Mining Using k-Core
  Analysis--Patterns, Anomalies and Algorithms}. In \bibinfo{booktitle}{{\em
  ICDM}}. \bibinfo{pages}{469--478}.
\newblock


\bibitem[\protect\citeauthoryear{{\v{S}}{\'\i}ma and Schaeffer}{{\v{S}}{\'\i}ma
  and Schaeffer}{2006}]%
        {vsima2006np}
\bibfield{author}{\bibinfo{person}{Ji{\v{r}}{\'\i} {\v{S}}{\'\i}ma} {and}
  \bibinfo{person}{Satu~Elisa Schaeffer}.} \bibinfo{year}{2006}\natexlab{}.
\newblock \showarticletitle{On the NP-completeness of some graph cluster
  measures}. In \bibinfo{booktitle}{{\em International Conference on Current
  Trends in Theory and Practice of Computer Science}}. Springer,
  \bibinfo{pages}{530--537}.
\newblock


\bibitem[\protect\citeauthoryear{Simon}{Simon}{1991}]%
        {simon1991partitioning}
\bibfield{author}{\bibinfo{person}{Horst~D Simon}.}
  \bibinfo{year}{1991}\natexlab{}.
\newblock \showarticletitle{Partitioning of unstructured problems for parallel
  processing}.
\newblock \bibinfo{journal}{{\em Computing systems in engineering\/}}
  \bibinfo{volume}{2}, \bibinfo{number}{2} (\bibinfo{year}{1991}),
  \bibinfo{pages}{135--148}.
\newblock


\bibitem[\protect\citeauthoryear{Steinhaeuser and Chawla}{Steinhaeuser and
  Chawla}{2008}]%
        {Steinhaeuser2008}
\bibfield{author}{\bibinfo{person}{Karsten Steinhaeuser} {and}
  \bibinfo{person}{Nitesh~V Chawla}.} \bibinfo{year}{2008}\natexlab{}.
\newblock \showarticletitle{{Community Detection in a Large Real-World Social
  Network BT - Social Computing, Behavioral Modeling, and Prediction}},
  \bibfield{editor}{\bibinfo{person}{Huan Liu}, \bibinfo{person}{John~J
  Salerno}, {and} \bibinfo{person}{Michael~J Young}} (Eds.).
  \bibinfo{publisher}{Springer US}, \bibinfo{address}{Boston, MA},
  \bibinfo{pages}{168--175}.
\newblock
\showISBNx{978-0-387-77672-9}


\bibitem[\protect\citeauthoryear{Sun and Han}{Sun and Han}{2013}]%
        {sun2013mining}
\bibfield{author}{\bibinfo{person}{Yizhou Sun} {and} \bibinfo{person}{Jiawei
  Han}.} \bibinfo{year}{2013}\natexlab{}.
\newblock \showarticletitle{Mining heterogeneous information networks: a
  structural analysis approach}.
\newblock \bibinfo{journal}{{\em SIGKDD Explorations\/}} \bibinfo{volume}{14},
  \bibinfo{number}{2} (\bibinfo{year}{2013}), \bibinfo{pages}{20--28}.
\newblock


\bibitem[\protect\citeauthoryear{Tang, Qu, Wang, Zhang, Yan, and Mei}{Tang
  et~al\mbox{.}}{2015}]%
        {line}
\bibfield{author}{\bibinfo{person}{Jian Tang}, \bibinfo{person}{Meng Qu},
  \bibinfo{person}{Mingzhe Wang}, \bibinfo{person}{Ming Zhang},
  \bibinfo{person}{Jun Yan}, {and} \bibinfo{person}{Qiaozhu Mei}.}
  \bibinfo{year}{2015}\natexlab{}.
\newblock \showarticletitle{LINE: Large-scale Information Network Embedding.}.
  In \bibinfo{booktitle}{{\em WWW}}.
\newblock


\bibitem[\protect\citeauthoryear{Tsourakakis, Pachocki, and
  Mitzenmacher}{Tsourakakis et~al\mbox{.}}{2017}]%
        {Tsourakakis2017}
\bibfield{author}{\bibinfo{person}{Charalampos~E Tsourakakis},
  \bibinfo{person}{Jakub Pachocki}, {and} \bibinfo{person}{Michael
  Mitzenmacher}.} \bibinfo{year}{2017}\natexlab{}.
\newblock \showarticletitle{{Scalable Motif-aware Graph Clustering}}. In
  \bibinfo{booktitle}{{\em WWW}}. \bibinfo{publisher}{International World Wide
  Web Conferences Steering Committee}, \bibinfo{address}{Republic and Canton of
  Geneva, Switzerland}, \bibinfo{pages}{1451--1460}.
\newblock
\showISBNx{978-1-4503-4913-0}


\bibitem[\protect\citeauthoryear{Van~Driessche and Roose}{Van~Driessche and
  Roose}{1995}]%
        {van1995improved}
\bibfield{author}{\bibinfo{person}{Rafael Van~Driessche} {and}
  \bibinfo{person}{Dirk Roose}.} \bibinfo{year}{1995}\natexlab{}.
\newblock \showarticletitle{An improved spectral bisection algorithm and its
  application to dynamic load balancing}.
\newblock \bibinfo{journal}{{\em Parallel computing\/}} \bibinfo{volume}{21},
  \bibinfo{number}{1} (\bibinfo{year}{1995}), \bibinfo{pages}{29--48}.
\newblock


\bibitem[\protect\citeauthoryear{Verma and Meila}{Verma and Meila}{2003}]%
        {Verma2003}
\bibfield{author}{\bibinfo{person}{Deepak Verma} {and} \bibinfo{person}{Marina
  Meila}.} \bibinfo{year}{2003}\natexlab{}.
\newblock \showarticletitle{{A comparison of spectral clustering algorithms}}.
\newblock \bibinfo{journal}{{\em University of Washington Tech Rep
  UWCSE030501\/}}  \bibinfo{volume}{1} (\bibinfo{year}{2003}),
  \bibinfo{pages}{1--18}.
\newblock


\bibitem[\protect\citeauthoryear{Voevodski, Teng, and Xia}{Voevodski
  et~al\mbox{.}}{2009}]%
        {voevodski2009finding}
\bibfield{author}{\bibinfo{person}{Konstantin Voevodski},
  \bibinfo{person}{Shang-Hua Teng}, {and} \bibinfo{person}{Yu Xia}.}
  \bibinfo{year}{2009}\natexlab{}.
\newblock \showarticletitle{Finding local communities in protein networks}.
\newblock \bibinfo{journal}{{\em BMC bioinformatics\/}} \bibinfo{volume}{10},
  \bibinfo{number}{1} (\bibinfo{year}{2009}), \bibinfo{pages}{297}.
\newblock


\bibitem[\protect\citeauthoryear{Xu, Ke, Wang, Cheng, and Cheng}{Xu
  et~al\mbox{.}}{2012}]%
        {Xu2012}
\bibfield{author}{\bibinfo{person}{Zhiqiang Xu}, \bibinfo{person}{Yiping Ke},
  \bibinfo{person}{Yi Wang}, \bibinfo{person}{Hong Cheng}, {and}
  \bibinfo{person}{James Cheng}.} \bibinfo{year}{2012}\natexlab{}.
\newblock \showarticletitle{{A Model-based Approach to Attributed Graph
  Clustering}}. In \bibinfo{booktitle}{{\em SIGMOD}} {\em
  (\bibinfo{series}{SIGMOD '12})}. \bibinfo{publisher}{ACM},
  \bibinfo{address}{New York, NY, USA}, \bibinfo{pages}{505--516}.
\newblock
\showISBNx{978-1-4503-1247-9}


\bibitem[\protect\citeauthoryear{Zhou, Cheng, and Yu}{Zhou
  et~al\mbox{.}}{2009}]%
        {Zhou2009}
\bibfield{author}{\bibinfo{person}{Yang Zhou}, \bibinfo{person}{Hong Cheng},
  {and} \bibinfo{person}{Jeffrey~Xu Yu}.} \bibinfo{year}{2009}\natexlab{}.
\newblock \showarticletitle{{Graph Clustering Based on Structural/Attribute
  Similarities}}.
\newblock \bibinfo{journal}{{\em VLDB\/}} \bibinfo{volume}{2},
  \bibinfo{number}{1} (\bibinfo{date}{aug} \bibinfo{year}{2009}),
  \bibinfo{pages}{718--729}.
\newblock
\showISSN{2150-8097}


\end{thebibliography}

\end{document}